\documentclass[twocolumn,aps,superscriptaddress,showpacs]{revtex4}
\usepackage{amsmath,amssymb}
\usepackage{graphicx}
\usepackage{verbatim}
\usepackage{amsfonts}
\usepackage{dcolumn}
\usepackage{bm}
\usepackage{color}
\usepackage[colorlinks,citecolor=blue]{hyperref}
\usepackage{framed} 
\usepackage{mathrsfs}
\usepackage{ntheorem}

\newcommand{\beq}{\begin{eqnarray} }
\newcommand{\eeq}{\end{eqnarray} }
\newcommand{\Beq}{\begin{eqnarray*} }
\newcommand{\Eeq}{\end{eqnarray*} }
\newcommand{\Bmat}{\left(\begin{matrix}}
\newcommand{\Emat}{\end{matrix}\right)}

\usepackage{color}
\usepackage[colorlinks,citecolor=blue]{hyperref}
\usepackage{amsmath,amssymb}
\usepackage{graphicx}
\usepackage{subfigure}
\usepackage{verbatim}
\usepackage{hyperref}
\usepackage{amsfonts}
\usepackage{dcolumn}
\usepackage{bm}
\usepackage{natbib}
\usepackage{mathrsfs}
\usepackage{empheq}
\usepackage{multirow}
 \usepackage[normalem]{ulem}
 \useunder{\uline}{\ul}{}

\newtheorem{theorem}{Theorem}
\newtheorem{lemma}{Lemma}

\newtheorem*{proof}{Proof}
\newtheorem{corollary}{Corollary}

\usepackage{longtable}
\usepackage{multirow}

\newcommand{\RNum}[1]{\uppercase\expandafter{\romannumeral #1\relax}}

\usepackage{ulem}

\begin{document}
\title{A Hamiltonian Approach for Obtaining Irreducible Projective Representations and the $k\cdot p$ Perturbation for Anti-unitary Symmetry Groups} 
\author{Zhen-Yuan Yang}
\affiliation{Department of physics, Renmin University, Beijing 100876, China.}

\author{Jian Yang} 
\affiliation{Beijing National Research Center for Condensed Matter Physics,
and Institute of Physics, Chinese Academy of Sciences, Beijing 100190, China}

\author{Chen Fang}
\affiliation{Beijing National Research Center for Condensed Matter Physics, and Institute of Physics, Chinese Academy of Sciences, Beijing 100190, China}

\author{Zheng-Xin Liu}
\email{liuzxphys@ruc.edu.cn}
\affiliation{Department of physics, Renmin University, Beijing 100876, China.}

\date{\today}
\begin{abstract}

As is known, the irreducible projective representations (Reps) of anti-unitary groups contain three different situations, namely, the real, the complex and quaternion types with torsion number 1,2,4 respectively. This subtlety increases the complexity in obtaining irreducible projective Reps of anti-unitary groups. In the present work, a physical approach is introduced to derive the condition of irreducibility for projective Reps of anti-unitary groups. Then a practical procedure is provided to reduce an arbitrary projective Rep into direct sum of irreducible ones. The central idea is to construct a hermitian Hamiltonian matrix which commutes with the representation of every group element $g\in G$, such that each of its eigenspaces forms an irreducible representation space of the group $G$. Thus the Rep is completely reduced in the eigenspaces of the Hamiltonian.  This approach is applied in the $k\cdot p$ effective theory at the high symmetry points (HSPs) of the Brillouin zone for quasi-particle excitations in magnetic materials. After giving the criterion to judge the power of single-particle dispersion around a HSP, we then provide a systematic procedure to construct the $k\cdot p$ effective model. 
\end{abstract}

\maketitle

\section{introduction}

Irreducible projective representations (IPReps) of groups, including the irreducible linear Reps as the trivial class of IPReps, play important roles in physics\cite{Schur,PRB2,PRB3,PRB4,PRB5,Slager2013,Barkeshli_2019}. In condensed matter physics, IPReps for discrete groups are widely used in obtaining selection rules or analyzing spectrum degeneracy\cite{GroupAppPhy}. For instance, in the band theory of itinerant electrons hopping in a crystal, the symmetry group is a space group whence the degeneracy of the energy spectrum at a momentum point is determined by the dimensions of IPReps of the little co-group\cite{Chen}. 

Owing to the importance of IPReps, it is urgent to judge if a Rep is reducible or not. For a finite unitary group $H$, a (projective) Rep $D(H)$ is irreducible if it satisfies the following condition, 
${1\over |H|} \sum_{h\in H} |\chi_h|^2=1,$
where $\chi^{(\nu)}(h)={\rm Tr} D^{(\nu)}(h)$ is the character of the element $h\in H$. When $D(H)$ is reducible, then ${1\over |H|} \sum_{h\in H} |\chi_h|^2 = \sum_{\nu} a_\nu^2$, where $a_\nu$ is the multiplicity of the irreducible Rep $(\nu)$ contained in $D(H)$. In this case, we need to transform it into a direct sum of irreducible Reps. The eigenfunction method\cite{Chen} is an efficient way of performing this reduction.

On the other hand, anti-unitary groups attract more and more interests. The well known Kramers degeneracy is a consequence of time-reversal symmetry for fermions with half-odd-integer spin. Time reversal also protects the helical gapless edge modes in topological insulators\cite{RMPKane,RMPZhang} or topological superconductors\cite{PRBRead,PRLZhang}. Especially, a large amount of materials in nature exhibit magnetic long-range order, the symmetries for some of these materials are described by anti-unitary groups called the magnetic space groups\cite{MathInSolids}, where the anti-unitary operations are generally combination of time reversal operation $T$ and certain unitary space-group element. The irreducible Reps (also called co-Reps) of the magnetic space groups are helpful to understand the properties of these materials. Especially, the low-energy physics of the quasi-particles at high symmetry points (HSPs) of the Brillouin zone (BZ) are characterized by the irreducible projective Reps of the little co-groups.

For anti-unitary groups, there are three types of irreducible Reps which are characterized by the torsion number.  Supposing that $M(G)$ is an irreducible Rep of an anti-unitary group $G$, and $H$ is the halving unitary subgroup $H\subset G$ with $G=H+T_0H$ ($T_0$ is anti-unitary). Then the torsion number is given by
$R={1\over|H|}\sum_{h \in H} |\chi(h)|^2,$ where $\chi(h)={\rm Tr}\ M(h)$ is the character of $h$. If $R=1$, the irreducible Rep $M(G)$ belongs to the real type; if $R=2$, then $M(G)$ belongs to the complex type; if $R=4$ then $M(G)$ belongs to the quaternion type\cite{CMP}. This subtlety  of anti-unitary groups increases the complexity in reducing an arbitrary projective Rep into the direct sum of irreducible ones, especially if some IPReps appea multipole times in the reucible Rep.

In the present paper, from a physical approach we derive the criterion to judge the irreducibility\cite{JMP} of a projective Rep $M(G)$ for a finite anti-unitary group $G$,
\Beq
{1\over |H|} \sum_{h\in H} {1\over2}\left[ \chi(h)\chi^*(h)  + {\rm Tr}[M(T_0h)M^*(T_0h)]\right] = 1, 
\Eeq
or equivalently
{\small
\begin{empheq}[box=\fbox]{align}\label{IR_condition}
{1\over |H|} \sum_{h\in H} {1\over2}\left[ \chi(h)\chi^*(h)  +\omega_2(T_0h, T_0h) \chi((T_0h)^2)\right] = 1, 
\end{empheq}
}where $\omega_2(T_0h, T_0h)$ is the factor system of the projective Rep.
In this approach, we consider Hermitian Hamiltonians in terms of single-particle bilinear operators which are commuting with all of the symmetry operations in $G$. If the only existing Hamiltonian is proportional to the identity matrix,  then the Rep $M(G)$ is irreducible. Otherwise, if there exist other linearly independent Hamiltonian, then $M(G)$ is reducible and the energies of the Hamiltonian can be used to distinguished each of the irreducible subspace. This provides an efficient method to reduce an arbitrary reducible Rep into a direct sum of irreducible ones. The advantage of the method is that no information of the irreducible Reps of the groups need to be known beforehand. We further generalize this approach to judge the power of the quasi-particle dispersions in magnetic semimetals, and then to obtain the $k\cdot p$ effective models \cite{Bardeen,Seitzkp} at the HSPs in the BZ.

The rest of the paper is organized as follows. In section \ref{sec:Crit}, we worm up by reviewing the IPReps of unitary groups, and then derive the formula (\ref{IR_condition}) for anti-unitary groups and interpret it in a physical Hamiltonian approach. In section \ref{sec:Reduction}, applying the Hamiltonian approach we provide the procedure to reduce an arbitrary Rep of finite groups (either unitary or anti-unitary) into a direct sum of IPReps.  In section \ref{sec:app}, we provide the criterion to judge if the degeneracy protected by IPReps of anti-unitary groups can be lift by certain perturbations or not, and then give the method to construct $k\cdot p$ effective Hamiltonian for magnetic materials. Section \ref{sec:sum} is devoted to the conclusions and discussions.  

Since any Rep of a finite group (no matter unitary or anti-unitary) can be transformed into a unitary one, in the present work we only discuss unitary Reps.

\section{A Hamiltonian approach: Condition for irreducible projective Reps}\label{sec:Crit}

\subsection{Unitary Groups}\label{sec: Untry}

Since the character of the identity Rep $(I)$ is $\chi^{(I)}(h)=1$ for any $h\in H$, the following quantity
\Beq
a_{(I)}^{(\nu\times \nu^*)}&=&{1\over |H|} \sum_{h\in H} |\chi^{(\nu)}_h|^2  \Big(\chi^{(I)}(h)\Big)^*\\
&=& {1\over |H|} \sum_{h\in H} {\rm Tr} [D^{(\nu)}(h)\otimes D^{(\nu)*}(h)],
\Eeq
stands for the multiplicity of the identity Rep appearing in the reduced Rep of the direct product $(\nu\times \nu^*)$, where ${(\nu)^*}$ is the complex conjugate of $(\nu)$. Then the condition of irreducibility of $(\nu)$ can be interpreted as the following: the direct product ${(\nu)\times (\nu)^*}$  contains only one identity Rep, namely $a_{(I)}^{(\nu\times \nu^*)}=1$. 

The expression $a_{(I)}^{(\nu\times \nu^*)}=1$ has a physical interpretation. Suppose the identical particle $\psi^\dag$ has $d$ internal components $\psi^\dag=(\psi_1^\dag,\psi_2^\dag, ..., \psi_{d}^\dag)$, which carries an Rep ${(\nu)}$ of the symmetry group $H$. This means that $\hat h\psi_i^\dag \hat h^{-1}= \sum_jD^{(\nu)}_{ji}(h)\psi_j^\dag$ or equivalently
\[
\hat h \psi^\dag \hat h^{-1} = \psi^\dag D^{(\nu)}(h).
\]
The hermitian conjugation gives $\hat h \psi \hat h^{-1} = [D^{(\nu)}(h)]^\dag \psi.$ The energy spectrum is described by the single-particle Hamiltonian
\beq\label{Ham}
\hat {\mathscr H}=\sum_i \psi_i^\dag \Gamma_{ij} \psi_j=\psi^\dag \Gamma \psi,
\eeq
where $\Gamma$ is an $d\times d$ matrix. The symmetry group $H$ means that the Hamiltonian is invariant under all the symmetry operations in the group $H$. In other words, for any $h\in H$, we have $\hat h \hat {\mathscr H}\hat h^{-1} = \hat {\mathscr H},$ which is equivalent to 
\beq\label{Gma}
D^{(\nu)}(h) \Gamma [D^{(\nu)}(h)]^\dag=\Gamma.
\eeq
Schur's lemma indicates that when $(\nu)$ is irreducible, then $\Gamma$ must be proportional to the identity matrix $\Gamma_0 \propto I$. If there exist another linearly independent matrix $\Gamma_1$ satisfying (\ref{Gma}), then it must have at least two eigenvalues. The eigenspace of each eigenvalue is closed under action of $H$ and hence form a Rep space of $H$. This means that the Rep $(\nu)$ is reducible. Therefore, if $I$ is the only one linearly independent matrix satisfying (\ref{Gma}), then the $d$-fold degenerate energy level of $\mathscr H$ cannot be lift and consequently $(\nu)$ is irreducible.

The equation (\ref{Gma}) can be expanded in the following form
\Beq
\sum_{j,k}\!D^{(\nu)}_{ij}\!(h) \Gamma_{jk} D^{(\nu)*}_{lk}\!(h) \!\!&=&\!\! \sum_{j,k}\! \left(\! D^{(\nu)}(h)\! \otimes \!D^{(\nu)*}(h) \! \right)_{il,jk}\!\! \Gamma_{jk}
\\ \!\!&=&\!\!\Gamma_{il}
\Eeq
for all $h\in H$.
If we reshape the matrix $\Gamma$ into an $d^2$-component  column vector (if the matrix $\Gamma$ is reshaped into the $d^2$-component vector column by column, then it should be transposed into $\Gamma^T$ before the reshaping), then this vector is the eigenvector of $D^{(\nu)}(h) \otimes D^{(\nu)*}(h)$ with eigenvalue 1 for all $h\in H$, $i.e.$ it carries the identity Rep of $H$. {\it In other words, the vector $\Gamma$ is the CG coefficient \cite{JMP2, JMP4} that combines the bases of $(\nu)$ and $(\nu)^*$ to a irreducible basis that belongs to the identity Rep $\chi^{(I)}(h)=1$. If $(\nu)$ is irreducible, then the CG coefficient is unique. }

Above discussion is valid no matter the Rep $(\nu)$ is  linear or projective. 

\subsection{Anti-unitary groups}\label{sec:anti}

In the following we generalize above approach to anti-unitary groups. Consider an anti-unitary group $G$ with $G=H+T_0 H$, where $H\in G$ is the halving unitary subgroup and $T_0$ is an anti-unitary element of the lowest order.
 
If $G$ is of type-I\cite{YangLiu}, namely, $T_0^2=E$, then $G$ is either a direct product group $G=H\times Z_2^T$  or a semi-direct product $G=H\rtimes Z_2^T$, where $Z_2^T=\{E,T_0\}$. If $T\in G$ (here $T$ is the time-reversal operation which commutes with all the other elements), then we choose $T_0=T$; otherwise, $T_0=u T$, where $u\notin G$ is a unitary operation satisfying $T_0^2=u^2=E$. 

On the other hand, if $G$ is of type-II, then $T_0^{2^n}=E$ with $n\geq 2$, hence $G$ cannot be written in forms of direct product or semi-direct product of a unitary group with $Z_2^T$. Obviously, the order of $T_0$ is at least 4 and $T_0^2\equiv\sigma$ is a unitary element in $H$, $\sigma\in H$. 

We consider an $d$-dimensional unitary projective Rep of $G$. Any element $g\in G$ is represented as $\hat g = M(g)K_{s(g)}$, which satisfies the relations $M^\dag(g)M(g)=I$ and
\Beq
M(g_1)K_{s(g_1)}M(g_2)K_{s(g_2)} = \omega(g_1,g_2)M(g_1g_2)K_{s(g_1g_{2})},
\Eeq
where $s(g)=1$, $K_{s(g)}=K$ if $g$ is anti-unitary and $s(g)=0$, $K_{s(g)}=I$ otherwise. The factor system $\omega_2(g_1,g_2)$ satisfies the cocycle equation $$\omega^{s(g_1)}(g_2,g_3) \omega^{-1}(g_1g_2,g_3) \omega(g_1,g_2g_3)\omega^{-1}(g_1,g_2) = 1.$$ Now we derive the condition for the irreducibility of $M(g)K_{s(g)}$.

\subsubsection{General Discussion}

Since unitary group elements are easier to handle, we expect that the irreducibility can be judged from the restrict Rep of the subgroup $H$. Noticing that $M(H)$ is possibly reducible even if $M(G)$ is irreducible, we have
\Beq \label{Htrace}
{1\over|H|}\sum_{h\in H} {\rm Tr} [M(h)\otimes M^*(h)]\geq 1.
\Eeq 
Actually $P^{(I)}={1\over|H|}\sum_{h\in H} M(h)\otimes M^*(h)$ is the projector onto the subspace of identity Reps contained in the direct product Rep $M(h)\otimes M^*(h)$. The eigenvalues of $P^{(I)}$ are either 1 (which occurs at least once) or 0, hence ${\rm Tr} P^{(I)}\geq1$.

We need to find a way to include the restrictions from the anti-unitary group elements. Adopting the physical argument as discussed in Sec. \ref{sec: Untry}, we consider a $d$-component particle $\psi^\dag$ which carries the (co-)Rep of $g\in G$, 
$$
\hat g\psi^\dag \hat g^{-1}= \psi^\dag M(g)K_{s(g)}.
$$
The Hamiltonian takes the same form of (\ref{Ham}), which is invariant under the action of all the group elements, $\hat g \hat {\mathscr H}\hat g^{-1}=\hat {\mathscr H}$, namely, 
\beq
&& M(h) \Gamma M(h)^\dag =\Gamma,\ \ \ \ \ \ \ h\in H \label{UH}\\
&& M(T_0) \Gamma^* M(T_0)^\dag=\Gamma. \label{T0}
\eeq
Similar to the discussion for unitary groups, the $\Gamma$ matrix can be considered the CG coefficient that combines the product Rep $M(g)\otimes M^*(g)K_{s(g)}, g\in G$ (a linear Rep) into the identity Rep. Since the identity matrix obviously satisfies the above two equations, the product Rep contains at least one identity Rep. We expect that the identity matrix is the unique linearly independent matrix satisfying (\ref{UH}) and (\ref{T0}) if the Rep $M(g)K_{s(g)}, g\in G$ is irreducible.

However, above statement is too strong for anti-unitary groups. We need one more constraint for $\Gamma$. Notice that if a matrix commutes with an irreducible (projective) Rep of an anti-unitary group, then this matrix may have two eigenvalues which are mutually complex conjugate to each other \cite{YangLiu}. To generalize the Schur's lemma to anti-unitary groups, the matrix $\Gamma$ needs to be Hermitian. Namely, if an Hermitian matrix commutes with the irreducible projective Reps of all the group elements of an anti-unitary group, then this matrix must be proportional to the identity matrix. 

Hence, in addition to (\ref{UH}) and (\ref{T0}), we should further require that
\beq\label{Herm}
\Gamma^\dag = \Gamma.
\eeq
If a non-hermitian matrix $\Gamma$ satisfies (\ref{UH}) and (\ref{T0}), then obviously its hermitian conjugate $\Gamma^\dag$ also does. Therefore, the linear combination $(\Gamma+\Gamma^\dag)$ is the required hermitian matrix \footnote{ The linear combinations $\Gamma_\pm = (\Gamma \pm \Gamma^\dag)/2$ still satisfy the relations (\ref{UH}) and (\ref{T0}). Now $\Gamma_\pm^\dag =\pm \Gamma_\pm$, meaning that $\Gamma_+$ is Hermitian and $\Gamma_-$ is anti-Hermitian. So $\Gamma_-$ violates (\ref{Herm}). On the other hand, if we transform $\Gamma_-$ into an hermitian matrix $i\Gamma_-$, then $M(T_0) (i\Gamma_-) ^* M^\dag (T_0) = - (i\Gamma_-)$, namely, the hermitian matrix $i\Gamma_-$ forms an eigenstate of $M(T_0)\otimes M^*(T_0)K$ with eigenvalue $-1$, which violates (\ref{T0}).}.

Therefore, when making using of the characters of the unitary subgroup $H$ to judge the irreducibility of $M(g)K_{s(g)}, g\in G$, we need a projection operator $P_{HT_0}=P_H P_{T_0}$ to project onto the  subspace formed by hermitian and $T_0$ symmetric matrices.  $P_{HT_0}$ is equivalent to project onto the eigenvectors of $M(T_0)\otimes M^*(T_0)K$ with eigenvalue 1 with the condition that the matrix form of these eigenvectors are hermitian. 

Therefore, considering (\ref{UH}), (\ref{T0}) and (\ref{Herm}), the irreducibility requires that
\beq \label{Irrep}
{\rm Tr} ( P^{(I)} P_{HT_0}) \!=\! {1\over|H|}\sum_{h\in H} {\rm Tr} [ M(h)\!\otimes\! M^*(h)P_{HT_0}]=1,
\eeq
namely, when projecting onto the hermitian and $T_0$ symmetric subspace, the identity Rep  only appears once in the direct product Rep $M(H)\otimes M^*(H)$.

Eq. (\ref{Irrep}) is a general expression of the criterion that a Rep of anti-unitary groups should meet if it is irreducible. However, the construction of the projection operator $P_{HT_0}$ is not straightforward. In the following we first consider a relatively simple case, $i.e.$ the type-I anti-unitary groups, and then generalize the conclusion to arbitrary anti-unitary groups.

\subsubsection{type-I anti-unitary groups}\label{TypeI}

For type-I anti-unitary groups with $T_0^2\equiv \sigma=E$, situations are much simpler.  For a unitary Rep, we have $M^\dag(T_0)M(T_0)=I$. On the other hand, $T_0^2=E$ indicates
$$
[M(T_0)K]^2 = M(T_0)M^*(T_0)= \eta_{0}\equiv \omega_2(T_0,T_0),
$$
where $\eta_0=\pm1$ is an invariant of the projective Rep of type-I anti-unitary groups. 

Under the hermitian condition (\ref{Herm}), the transpose of (\ref{T0}) yields $ \Gamma^T = M^*(T_0) \Gamma M^T(T_0)$, namely 
$$
M(T)\Gamma^T=M(T_0)M^*(T_0) \Gamma M^T(T_0) = \eta_0 \Gamma M^T(T_0).
$$ 
Defining $\tilde \Gamma=\Gamma M^T(T_0),$ then we have
\beq\label{ST0}
\tilde\Gamma^T=\eta_{0}\tilde \Gamma.
\eeq
This means that $\tilde \Gamma$ is either symmetric (if $\eta_0=1$) or anti-symmetric (if $\eta_0=-1$). This symmetry condition is a direct consequence of the anti-unitary symmetry condition (\ref{T0}) and hermiticity condition (\ref{Herm}). In the following, we will say $\tilde\Gamma$ to be $\eta_0$-symmetric if it satisfies  (\ref{ST0}). 

Since $\Gamma=\tilde \Gamma  M^*(T_0)$, we rewrite the Hamiltonian as
\[
\hat {\mathscr H} = \psi^\dag \tilde \Gamma  M^*(T_0)\psi = \psi^\dag \tilde \Gamma \tilde \psi,
\]
then the basis $\psi$ undergoes a unitary transformation $\psi\to \tilde \psi =M^*(T_0)\psi$. Under the action of $h\in H$, $\tilde \psi$ vary as $\hat h \tilde \psi \hat h^{-1}= M^*(T_0) M^\dag(h)\psi=M^*(T_0)M^\dag (h)M^T(T_0)\tilde \psi$. 
For convenience, we define the following Rep for $h\in H$,
\beq\label{Fh}
F(h)&=&M(T_0)M^*(h)M^\dag(T_0), 
\eeq
which is equivalent to $M^*(h)$ with ${\rm Tr\ } F(h) = {\rm Tr\ } M^*(h)=\chi^*(h)$. Accordingly, $\tilde \psi$ vary as $\hat h \tilde \psi \hat h^{-1}= F^T(h)\tilde \psi$. Hence, the condition $\hat h \hat {\mathscr H} \hat h^{-1} = \hat {\mathscr H}$ requires that
\beq\label{Hinv}
M(h) \tilde \Gamma F^T(h) &=& \tilde \Gamma,
\eeq
which is the deformation of (\ref{UH}). Similarly, (\ref{T0}) is transformed into 
\beq\label{T0tG}
M(T_0)\tilde\Gamma^*M^T(T_0)=\tilde\Gamma.
\eeq

As before, $\tilde \Gamma$ can be considered as the CG coefficient that couples the direct product Rep $ V(g) = M(g)\otimes F(g)K_{s(g)}, g\in G$ (a linear Rep) to the identity Rep, namely
\beq\label{eigMF}
\sum_{ij} V_{kl,ij}(g)K_{s(g)} \tilde \Gamma_{ij} &=& \sum_{ij} M_{ki}(g) F_{lj}(g) K_{s(g)}\tilde \Gamma_{ij} \notag\\
&=&\tilde\Gamma_{kl}K_{s(g)},
\eeq
with $V(T_0)=M(T_0)\otimes M(T_0)$. If the CG coefficient matrix $\tilde\Gamma$ satisfies the $\eta_0$-symmetry condition (\ref{ST0}), then we only need to consider the unitary elements $h\in H$.

Obviously, $\tilde \Gamma=M^T(T)$  ({\it i.e.} $\Gamma=I$) satisfies the relations (\ref{ST0}) and (\ref{Hinv}). The irreducibility of $M(g)K_{s(g)}, g\in G$  indicates that there is a unique linearly independent solution. In other words, when projected onto the $\eta_0$-symmetric subspace by the projection operator $P_{\eta_0}$, the product Rep $V(h) = M(h)\otimes F(h)$ only contains a single identity Rep, 
\beq\label{Irrep-I}
a_{(I)}^{[M\otimes F]_{\eta_0}} = {1\over |H|} \sum_{h\in H}{\rm Tr}\left[ M(h)\otimes F(h) P_{\eta_{0}}\right]=1.\ \ \ 
\eeq

To obtain the matrix form of $[M(h)\otimes F(h) P_{\eta_{0}}]$, we devide $M(h)\otimes F(h)$ into two parts,
\Beq
\sum_{ij}\big(M(h)\!\!\!\!\!\!\!\!\!\!\!\! &&\otimes F(h)\big)_{kl,ij} \tilde \Gamma_{ij}\\
 &= &\sum_{ij}{1\over2}\big[M_{ki}(h) F_{lj}(h) + \eta_0M_{kj}(h) F_{li}(h) \big]\tilde \Gamma_{ij}+\\
&& \sum_{ij}{1\over2}\big[M_{ki}(h) F_{lj}(h) - \eta_0M_{kj}(h) F_{li}(h) \big]\tilde \Gamma_{ij}.
\Eeq
Noticing that the second summation on the righthand side vanishes owing to $\tilde\Gamma_{ij}=\eta_0\tilde\Gamma_{ji}$, so we have
\beq\label{MFeta0}
[V_{\eta_{0}} (h)]_{kl,ij}= \frac{1}{2}  \big(  M_{ki}(h) F_{lj}(h)  +  \eta_{0} M_{kj}(h) F_{li}(h) \big),
\eeq
where we have used the notation $V_{\eta_0}(h) = [M(h)\otimes F(h)P_{\eta_0}]$.

Introducing the unit twist matrix 
\[
(\mathscr T)_{kl,ij}=\delta_{kj}\delta_{li}
\]
with $(X\mathscr T)_{kl,ij}=X_{kl,ji}$ for an arbitrary matrix $X$, then the projection operator $P_{\eta_0}$ can be expressed as 
$$P_{\eta_0}={1\over2}(I+\mathscr T_{\eta_0}),$$ 
where $I$ is the identity matrix and $\mathscr T_{\eta_0}=\eta_0\mathscr T$. Hence $V_{\eta_0}(h)$ in (\ref{MFeta0}) can be written as 
\[
V_{\eta_0}(h) = \big( M(h)\otimes F(h)\big) P_{\eta_0} = {1\over2}V(h)(I+\mathscr T_{\eta_0}).
\]

Although $V_{\eta_0}(h)$ does not form a Rep of $H$,  the common eigenvector of $V_{\eta_0}(h)$ with eigenvalue 1 does carries the identity Rep of $H$ (see Theorem \ref{Theorem1} for the special case in which $\sigma=E$). Defining $P^{(I)}={1\over|H|} \sum_{h\in H} V(h)$, above statement indicates that if $\tilde\Gamma$ satisfies 
\beq\label{VGamma}
P^{(I)}P_{\eta_0} \tilde\Gamma_{ij}=\tilde\Gamma_{kl},
\eeq
then it simultaneously satisfies the relations (\ref{ST0}) and (\ref{Hinv}). Furthermore, by choosing proper bases in the supporting space of $P^{(I)}P_{\eta_0}$, the $T_0$-symmetry condition (\ref{T0tG}) and finally the hemiticity condition (\ref{Herm}) can be ensured (see Appendix \ref{app:Peta0} for details). Thus the criterion (\ref{Irrep-I}) for the irreducibility is valid.

From the matrix form in  (\ref{MFeta0}),  the criterion (\ref{Irrep-I}) can be expressed in terms of the characters $\chi(h)={\rm Tr\ }M(h)$ of the unitary elements $h\in H$, namely
\Beq
{1\over |H|} \sum_{h\in H} {1\over2}\Big( \chi(h) \chi^*(h) + \eta_0 {\rm Tr\ }[F(h)M(h)] \Big) = 1,
\Eeq
where ${\rm Tr\ }F(h)=\chi^*(h)$ has been used. Furthermore, by denoting  $\bar h=T_0^{-1}hT_0=T_0hT_0$, we have
\Beq
F(h)&=&M(T_0)M^*(h)M^\dag(T_0)\nonumber\\
&=&\omega_2(T_0,h)\omega_2(T_0h,T_0)\omega_2(T_0,T_0) M(\bar h).
\Eeq
Above can be further simplified using the cocycle relation $\omega_2^{-1}(T_0,h)\omega_2^{-1}(\bar h,h)\omega_2(T_0h,T_0h)\omega_2^{-1}(T_0h,T_0)=1$, 
which yields $\omega_2(T_0,h)\omega_2(T_0h,T_0)={\omega_2(T_0h,T_0h)\over \omega_2(\bar h,h)}$. Noticing that $\eta_0^2=1$, therefore we have
\Beq
\eta_0 {\rm Tr\ }[F(h) M(h)]&=& {\omega_2(T_0h,T_0h)\over \omega_2(\bar h,h)} {\rm Tr\ }[M(\bar h) M(h)]\\
&=& \omega_2(T_0h,T_0h) {\rm Tr\ }[M(\bar hh)]\\
&=&\omega_2(T_0h,T_0h) \chi ((T_0h)^2)\\
&=&{\rm Tr\ } [M(T_0h)M^*(T_0h)].
\Eeq

Finally, we reach the simplified irreducible condition 
\begin{empheq}[box=\fbox]{align} \label{Anti}
{1\over 2|H|} \sum_{h\in H}\! \Big( \chi(h)\chi^*(h)  +\omega_2(T_0h, T_0h) \chi\big((T_0h)^2\big)\Big) = 1. 
\end{empheq}
Above expression is independent on the gauge choice of the projective Rep. The factor system $\omega_2(T_0h, T_0h)$ can be avoided by the replacement $\omega_2(T_0h, T_0h) \chi\big((T_0h)^2\big) = {\rm Tr\ } [M(T_0h)M^*(T_0h)]$.

In the following we show that above condition of irreducibility also works for type-II anti-unitary groups. 

\subsubsection{type-II anti-unitary groups}

For type-II anti-unitary groups, we denote $T_0^2\equiv\sigma$. Similar to previous discussion, if we define $\tilde \Gamma=\Gamma M^T(T_0)$, 
then (\ref{Hinv}) and (\ref{T0tG}) are the conditions $\tilde \Gamma$ should satisfy. Furthermore, recalling $\Gamma$ is hermitian, we have
\begin{align}\label{GeneralSymConduction}
\tilde \Gamma^T = M(T_0)\Gamma^T&=M(T_0)M^*(T_0) \Gamma M^T(T_0) \notag\\
&=\eta_{0}M(\sigma)\tilde \Gamma.
\end{align}
The self-consistency condition $(\tilde \Gamma^T)^T=\tilde \Gamma$ requires that $\tilde\Gamma=M(\sigma)\tilde \Gamma F^T(\sigma)$ with $F(\sigma)=\eta_0^2M(\sigma)$, or equivalently
$$V(\sigma)\tilde\Gamma = M(\sigma)\otimes F(\sigma)\tilde\Gamma = \tilde\Gamma.$$ 
This requirement is guaranteed if (\ref{Hinv}) is satisfied.

Therefore, the $T_0$ symmetry condition (\ref{T0}) and the hermiticity condition (\ref{Herm}) combine to a single restriction (\ref{GeneralSymConduction}), i.e. $\tilde \Gamma=[\eta_0M(\sigma) \tilde \Gamma]^T$. Now we define a generalized twist operator $\mathscr T_{\eta_0}$ which transforms $\tilde\Gamma$ into $[\eta_0M(\sigma) \tilde \Gamma]^T$,
\beq\label{scrTeta0}
\mathscr T_{\eta_0} = \eta_0 \mathscr T [M(\sigma)\otimes I ].
\eeq
It holds that $(\mathscr T_{\eta_0})^2= M(\sigma)\otimes F(\sigma)=V(\sigma)$. Thus $\mathscr T_{\eta_0}$ defines a generalized `transpose' of $\tilde\Gamma$ given that $V(\sigma)\tilde\Gamma=\tilde\Gamma$ is satisfied. 

When projected to the eigenspace of $V(\sigma)$ with eigenvalue 1,  the operator $\mathscr T_{\eta_0}$ has eigenvalues $\pm1$. Therefore the projector onto the generalized $\eta_0$-symmetric subspace is given by 
\Beq
P_{\eta_0}= {1\over2}P_{\sigma} (I + \mathscr T_{\eta_0}), 
\Eeq
where $P_{\sigma}$ is a projection onto the eigenspace of $V(\sigma)$ with eigenvalue 1. Defining the projector onto the subspace of identity Reps as $P^{(I)}={1\over |H|}\sum_{h\in H} V(h)$, then it is obvious that $P^{(I)}P_\sigma=P^{(I)}$. Hence the irreducibility condition (\ref{Irrep-I}) can be written as 
\begin{align}\label{typeII:criterion}
a_{(I)}^{[M\otimes F]_{\eta_0}} = {\rm Tr }(P^{(I)}P_{\eta_0})={1\over2}{\rm Tr }\big(P^{(I)}(I+\mathscr T_{\eta_0})\big)=1.
\end{align}
After some calculations, above criterion of irreducibility can be simplified to the same form as (\ref{Anti}) (see Appendix \ref{app:proof}).

From the definition of torsion number of irreducible Reps and the equation (\ref{Anti}), one can easily verify the following relation for any anti-unitary group $G=H+T_0H$,
\Beq
{1\over |H|} \!\!\sum_{h\in H} \!\!\omega_2(T_0h, T_0h) \chi\big((T_0h)^2\big) \!\!&=&\!\! {1\over |H|}\!\! \sum_{u\in T_0\!H}\!\!\!{\rm Tr} [M(u)M^*u)]
\\ \!\!&=&\!\! \left\{ \begin{aligned}  &1,&{\ \ \ \ \rm if}\ R=1 \\ &0,& {\ \ \ \ \rm if}\ R=2\\&-1,& {\ \ \ \ \rm if}\ R=4 \end{aligned}\right. ,
\Eeq
which provides another way to obtain the torsion number.

\section{Hamiltonian approach for the Reduction of Projective Reps} \label{sec:Reduction} 
The criterion of judging the irreducibility actually provides a practical procedure to reduce reducible projective Reps of finite groups. In the following, we discuss unitary groups and anti-unitary groups separately. \\

\subsection{Reduction of Reps for Unitary groups}\label{sec:UnitRed}

For a general hermitian Hamiltonian matrix $\Gamma$ satisfying (\ref{Gma}), each of its eigenspace is an irreducible subspace of the unitary group $H$. Namely, the eigenvalues of $\Gamma$ can be used to label the irreducible projective Reps of $H$. In order to simultaneously block diagonalize the restrict Rep of $H$ and its subgroups, we can make use of the class operators of $H$ and those of its subgroups\cite{Chen} to lift the degeneracy of $\Gamma$.  

Therefore, the central step is to construct the hermitian Hamiltonian matrix $\Gamma$. Here we summarize the reduction procedure in the following three steps:\\
\indent(1) Obtain the subspace $\mathcal L^{(I)}$ which carries the identity Reps of $M(H)\otimes M^*(H)$, namely, find all the bases $v^{(I)i}\in \mathcal L^{(I)}$ such that for any group element ${ h\in H}$,
\Beq
M(h)\otimes M^*(h) v^{(I)i} = v^{(I)i}; 
\Eeq
\indent(2) Chose an arbitrary basis $v=\sum_i r_i v^{(I)i} \in \mathcal L^{(I)}$, where $r_i\in \mathbb R$ are arbitrary real numbers, reshape $v$ into a matrix $\Gamma_0$, namely
\Beq
(\Gamma_0)_{ab} = v_{a(N-1)+b};  
\Eeq
and then construct an hermitian matrix $\Gamma = (\Gamma_0+\Gamma_0^\dag) + i (\Gamma_0-\Gamma_0^\dag)$;\\
\indent(3) Diagonalize the class operators $C$ of $M(H)$, and the class operators $C(s)$ of its subgroup chain $H_1\subset H_2\subset...\subset H$, and the matrix $\Gamma$ simultaneously,
\Beq
\Bmat C\\C(s)\\ \Gamma \Emat \phi^{(\nu)_{\varepsilon}}_m = \Bmat \nu\\ m\\ \varepsilon\Emat \phi^{(\nu)_{\varepsilon}}_m,
\Eeq 
then the eigenvectors $\phi^{(\nu)_{\varepsilon}}_m$ are the irreducible bases. The eigenspace of `energy' $\varepsilon$ is an irreducible Rep space, thus we can use the energy $\varepsilon$ to label the multiplicity $(\nu)_{\varepsilon}$ if the IPRep $(\nu)$ occurs more than once. The class operators are defined as the following\cite{Chen,YangLiu}
\beq\label{Ci}
C_i = \sum_{h_a\in H} M(h_a) M(h_i) M^\dag(h_a), 
\eeq
and $C$ is a linear combination of $C_i$ with $C=\sum_i r_i C_i$  where $r_i\in \mathbb R$ are arbitrary real numbers. The operators $C(s)$ are defined in a similar way, which are used to lift the degeneracy of the eigenvalues and to reduce the restricted Reps of the subgroups on the chain $H_1\subset H_2\subset...\subset H$.

In the first step, the eigenvectors of $M(h)\otimes M^*(h), h\in H$ with eigenvalue 1 are required. When the dimension $N$ of $M(h)$ is large, it seems that one need to solve the eigenstates of matrices with dimension $N^2$. Actually, this complexity can be avoided in two ways.  
  
One way is to obtain the eigenvectors of $M(h)\otimes M^*(h)$ from the eigenstates of $M(h)$ and $M^*(h)$. Since the eigenvalues of $M(h)\otimes M^*(h)$ are the product of the eigenvalues of $M(h)$ and $M^*(h)$, the eigenvectors of the product matrix with eigenvalue 1 is the direct product of the eigenstates of $M(h)$ and $M^*(h)$ whose eigenvalues are mutually complex conjugate. For all the elements $h\in H$ we can construct the eigenspace of $M(h)\otimes M^*(h)$ with eigenvalue 1 in the same way, then any state in the intersection of such eigenspaces satisfies the condition (1).

The other way is to construct the matrix $\Gamma_0$ directly, 
$$
\Gamma_0=\sum_{h\in H} M(h)AM^{\dag}(h),
$$
where $A$ is an arbitrary square matrix\footnote{ If we regard $A$ as a vector, then the above equation can be written as $(\Gamma_0)_{ij} = \sum_{h\in H}  M_{ik}(h)M_{jl}^*(h)A_{kl} = |H|P_{ij,kl}^{(I)}A_{kl}$, where $P^{(I)}={1\over |H|}\sum_{h\in H}M(h)\otimes M^*(h)$ is the projection operator projecting onto the subspace of identity Reps contained in the direct product Rep $M(H)\otimes M^*(H)$.}.  Obviously above $\Gamma_0$ satisfies the commutation relation $M(h)\Gamma_0 M^\dag (h)=\Gamma_0$, which is equivalent to the eigen problem $M(h)\otimes M^*(h)v= v$ with the vector $v$ reshaped from $\Gamma_0$. Therefore, thus constructed matrix satisfies the conditions in step (1) and step (2). Practically  this method is  more straightforward.

\subsection{Reduction of Reps for anti-unitary groups}\label{sec:redau}
The same idea can be generalized to reduce general Reps of anti-unitary groups $G=H+T_0H$. For a general matrix $\Gamma$ satisfying  (\ref{UH}), (\ref{T0}) and (\ref{Herm}), each of its eigenspaces is an irreducible projective Rep space of $G$. To lift the degeneracy of the eigenvalues of $\Gamma$ , we can make use of the class operators of $H$ and those of its subgroups.  

The central step is to construct the hermitian Hamiltonian matrix $\Gamma$ satisfying the restrictions (\ref{UH}), (\ref{T0}) and (\ref{Herm}). We summarize the reduction procedure as the following:\\
\indent(1) Following the method in section \ref{sec:UnitRed}, obtain a matrix $\Lambda_0$ which is commuting with $M(h), h\in H$, and then construct a hermitian matrix $\Lambda=(\Lambda_0 + \Lambda_0^\dag)  + i(\Lambda_0 - \Lambda_0^\dag)$;\\
\indent(2) Construct a matrix $\Gamma$ from $\Lambda$ 
\Beq
\Gamma &=&  \Lambda  + M(T_0)K \Lambda KM^\dag(T_0)\\
&=&\Lambda  + M(T_0) \Lambda^* M^\dag(T_0).
\Eeq
It is easily verified that $\Gamma M(T_0)K = M(T_0)K\Gamma$ because $T_0^2\in H$,  $M(T_0^2)\Lambda =\Lambda M(T_0^2)$, and that $[M(T_0)K]^2 =\omega_2(T_0,T_0)M(T_0^2)$. Furthermore, noticing that $hT_0=T_0(T_0^{-1}hT_0)$ and that $(T_0^{-1}hT_0)\in H$, it can be shown that $\Gamma$ commutes with $M(h)$ for all $h\in H$; \\
\indent(3) Simultaneously diagonalize the class operators  $C$ of $M(H)$ [see (\ref{Ci}) for definition], the class operators $C(s)$ of the subgroup chain $H_1\subset H_2\subset...\subset H$, and the Hamiltonian matrix $\Gamma$,
\Beq
\Bmat C\\C(s)\\ \Gamma \Emat \phi^{(\nu)_\varepsilon}_m = \Bmat \nu\\ m\\ \varepsilon\Emat \phi^{(\nu)_\varepsilon}_m,
\Eeq 
then the eigenvectors $\phi^{(\nu)_\varepsilon}_m$ are the irreducible bases, where the bases with the same `energy' $\varepsilon$ belong to the same irreducible Rep-space $(\nu)_\varepsilon$. 

If IPReps with torsion number $R=4$ are contained in $M(G)$ after the reduction, then the restricted Rep of $H$ in each of the $R=4$ IPRep is a direct sum of two identical copies of irreducible Reps of $H$. However, both $\Gamma$ and $C$ can only provide a single eigenvalue in the IRRep of $G$. Therefore the quantum number $m$ in step (3) are doubly degenerate. In this case,  we can use the hermitian matrix $\Lambda$ to distinguish the two identical irreducible Reps of $H$. It is obvious that $\Lambda$ commutes with $\Gamma, C$ and $C(s)$, so 
we can add it to the commuting operators in step (3), 
\Beq
\Bmat C\\\Lambda, C(s)\\ \Gamma\Emat \phi^{(\nu)_{\varepsilon}}_{\varepsilon_H,m} = \Bmat \nu\\ \varepsilon_H, m\\ \varepsilon\Emat \phi^{(\nu)_{\varepsilon}}_{\varepsilon_H, m},
\Eeq 
then all the degeneracies are lifted. 

Notice that we have used the class operators of $H$ to define the class operator $C$ for simplicity. The eigenvalues $\nu$ are not necessarily real (it is not real if $R=2$). In this case the eigenspaces of $\nu$ and $\nu^*$ belong to the same IPRep of $G$. One can also adopt the class operators of the total group $G$, $C_{i+} = C_{h_i} + C_{T_0h_iT_0^{-1}} + C_{h_i^{-1}} + C_{T_0h_iT_0^{-1}}, C_{i-} = i(C_{h_i} + C_{T_0h_iT_0^{-1}} - C_{h_i^{-1}} - C_{T_0h_iT_0^{-1}})$ to construct $\mathcal C= \sum_i(r_{i+} C_{i+} + r_{i-} C_{i-})$ \cite{YangLiu}, where $C_{h_i}$ is the class operator of $H$ in the restricted Rep and $r_{i\pm}\in\mathbb R$ are real numbers.  Then the eigenvalues of $\mathcal C$ are always real numbers, but in this case the operators $\mathcal C(s)$ should include the class operators of $H$ and those of its subgroups.\\

\section{Application of the Hamiltonian approach in perturbation theory}\label{sec:app}

The Hamiltonian approach can be generalized to obtain the response of the system to symmetry breaking probe fields if the low-energy physics is dominated by particle-like excitations, such as the electron-like quasiparticles in metals, Bogoliubov quasi-particles in superconductors or the magnon excitations in the spin sector. We restrict our discussion to irreducible projective Reps of anti-unitary groups. 

\subsection{$ k\cdot  p$ perturbation around high symmetry points}

In this section, we discuss the nodal-point and nodal-line structures in magnetic materials whose symmetry group are either type-III or type-IV Shubnikov magnetic space groups. The symmetry operations which keep a momentum $\pmb k$ invariant (up to a reciprocal lattice vector) form a magnetic point group $G_0(\pmb k)$ which is called the little co-group. The degeneracy of the energy bands at $\pmb k$ is determined by the irreducible (projective) Reps of the little co-group. The dispersion around $\pmb k$ can be obtained using the $ k\cdot  p$ perturbation theory. 

Suppose that the little co-group $G_0(\pmb k)$ has a $d$-dimensional irreducible (projective) Rep, which is carried by the quasi-particle bases $\psi_{\pmb k}^\alpha,\alpha=1,2,...,d$,  with
\beq
\hat g \psi_{\pmb k}^\dag \hat g^{-1}&=&\psi_{\pmb k}^\dag M(g)K_{s(g)},\label{kdotp_transdag}\\
\hat g \psi_{\pmb k}\hat g^{-1}&=&K_{s(g)}M(g)^\dag \psi_{\pmb k},\label{kdotp_trans}
\eeq
for $g\in G$.  The degeneracy is generally lifted at the vicinity of $\pmb k$. When $\delta\pmb k$ is small enough,  it is expected that $\psi_{\pmb{k}+\delta \pmb k}^\dag $ and $\psi_{\pmb{k}+\delta \pmb k}$ vary in the way similar to (\ref{kdotp_transdag}) and (\ref{kdotp_trans}) under the group action,
\beq
\hat g \psi_{\pmb{k}+\delta \pmb k }^\dag \hat g^{-1}&=&\psi_{\pmb{k}+\hat g\delta \pmb k}^\dag M(g)K_{s(g)},\label{hdk}\\
\hat g \psi_{\pmb{k}+\delta \pmb k}\hat g^{-1}&=&K_{s(g)}M(g)^\dag \psi_{\pmb{k}+\hat g \delta \pmb k}.\label{Tdk}
\eeq
Suppose the Hamiltonian at $\pmb k +\delta \pmb k$ is given by
\begin{align}\label{kdotp_ham_deltak} 
H_{\pmb{k}+\delta \pmb k}=\psi_{\pmb{k}+\delta \pmb k}^\dag \Gamma (\delta\pmb k) \psi_{\pmb{k}+\delta \pmb k},
\end{align}
where $ \Gamma (\delta\pmb k)$ is an Hermitian matrix $ \Gamma^\dag (\delta\pmb k)= \Gamma (\delta\pmb k)$. 
When summing over all the momentum variation, the total Hamiltonian should preserve the $G$ symmetry, $i.e.$, 
\begin{align}\label{total_Ham}
\hat g \left(\sum_{\delta k}H_{\pmb{k}+\delta \pmb k}\right)\hat g^{-1} =\left(\sum_{\delta k}H_{\pmb{k}+\delta \pmb k}\right), 
\end{align}
for all $g\in G$. Substituting the equations (\ref{kdotp_ham_deltak}), (\ref{hdk}) and (\ref{Tdk}) into (\ref{total_Ham}), we obtain,
\beq\label{Gammak}
M(g)K_{s(g)} \Gamma(g^{-1}\delta \pmb k) K_{s(g)} M^\dag (g) = \Gamma(\delta\pmb  k).
\eeq
which is the most general symmetry requirement.

If the leading order of $\Gamma(\delta\pmb k)$ is linear in $\delta\pmb k$, namely,  $\Gamma(\delta\pmb k)\sim \delta \pmb k\cdot \pmb\Gamma$, then the dispersion around this high-degeneracy point froms a cone. For fermionic systems, a conic dispersion is called a Dirac cone\cite{KaneDirac,Tang,XuPRB,Armitage,HUA,Watan,Cano,XuNature,elcoro2020magnetic,Bouhon2020} if $d=4$ and if $\tilde T=\mathcal I T$  ($\mathcal I$ is the spacial inversion operation) is an element of $G_0(\pmb k)$ such that the energy bands are doubly degenerate away from $\pmb k$. On the other hand, if the degeneracy remains unchanged along a special line crossing the point $\pmb k$, then this line is called a nodal line\cite{NodalLine1,Burkov2011, NodalLine3, NodalLine2,Fang2015, Weng2015,NodalLine4, Gei2019, guo2020, cui2020,YFL}.

Following the idea of the previous sections, here we provide a criteria to judge whether the dispersion around the point $\pmb k$ is linear or of higher order, and whether the degeneracy is stable in a high symmetry line. 

\subsubsection{Nodal points with linear dispersion} \label{sec:disper_kpHam}

Firstly, we consider linear dispersion around $\pmb k$, namely, 
\beq\label{linear}
\Gamma(\delta\pmb k) =  \sum_{m=1}^3 \delta k_m  \Gamma^m + O(\delta k^2).
\eeq

Here $\delta \pmb k$ is a dual vector under the point group operations in $H$, namely,
\beq
&&\hat h\delta k_m = \sum_n D^{(\bar v)}_{mn}(h)\delta k_n,\label{hk} 
\eeq
where $(\bar v)$ is the dual Rep of the vector Rep  $(v)$ of the unitary subgroup $H$ with $D^{(\bar v)}(h)=\left([D^{(v)}(h)]^{-1}\right)^T$. The vector Rep is real, so $(\bar v)$ is equivalent to $(v)$ [in orthonormal bases, $(\bar v)$ is identical to $(v)$, but we do not require the bases $[\pmb b_1,\pmb b_2,\pmb b_3]$ in the reciprocal space to be orthonormal]. 

From (\ref{Gammak})$\sim$(\ref{hk}), it can be shown (see Appendix \ref{app:vec}) that $\pmb\Gamma$ carries the dual vector Rep of $H$, namely,
\begin{align}
&M(h)\Gamma^m M(h)^\dag=\sum_n D^{(\bar v)}_{nm}(h)\Gamma^n. \label{kdotp_Gampro_F1}
\end{align}
In the following we first assume that the vector Rep $(v)$ is irreducible. The case $(v)$ is reducible will be mentioned later.

According to the action of $T_0$ on $\delta \pmb k$, we first discuss a special case where $T_0$ acts trivially on $\delta \pmb k$, then go to the general cases.

\section*{The special case $T_0 \delta \pmb{k}=\delta \pmb{k}$}\label{kdotp_seca}
Firstly we consider the case that $T_0$ acts trivially on $\delta \pmb k$, 
\beq
&&T_0 \delta \pmb{k}=\delta \pmb{k}.\label{Tk}
\eeq
From (\ref{Gammak}), (\ref{linear}) and above equation, we have,
\begin{align}
M(T_0)K\Gamma^mKM(T_0)^\dag=\Gamma^m. \label{kdotp_Gampro_F2}
\end{align}
The requirements (\ref{kdotp_Gampro_F1}) and (\ref{kdotp_Gampro_F2}) are similar to (\ref{UH}) and (\ref{T0}), respectively. If there exists three $d\times d$ Hermitian matrices $\Gamma^{1,2,3}$ satisfying these requirements, then the dispersion around $\pmb k$ forms a cone. From the discussion in \ref{sec:anti},  we can judge the existence of $\Gamma^{1,2,3}$  by checking if the projected space $M(h)\otimes M^*(h)P_{HT_0}$ (or equivalently the projected space $M(h)\otimes F(h)P_{\eta_{0}}$) contains the dual vector Rep $(\bar v)$ of $H$.

When the vector Rep $(v)$ of $H$ is irreducible, then the existence of linear dispersion can be checked by calculating the following quantity,
{\footnotesize
\begin{empheq}[box=\fbox]{align} \label{kdotp_judge_Final}
a_{(\bar v)}^H=\!{1\over2 |H|}\! \sum_h \!\left[ \chi(h)\chi^*(h) \! +\!\omega(hT_0, hT_0) \chi((hT_0)^2)\right] \chi^{(v)}(h),
\end{empheq}
}where $[\chi^{(\bar v)}(h)]^* =\chi^{(\bar v)}(h)= \chi^{(v)}(h)$ has been used. If $a_{(\bar v)}^H$ is a nonzero integer, then the dispersion is linear along all directions. 

The existence of $\Gamma^{1,2,3}$ under the conditions (\ref{kdotp_Gampro_F1}) and (\ref{kdotp_Gampro_F2}) can also be checked straightforwardly by reducing the product Rep $M(g)\otimes M^*(g)K_{s(g)}, g\in G$ into direct sum of IPReps using the method introduced in section \ref{sec:redau}. If the resultant IPReps contain the dual vector Rep(s) whose bases are hermitian when reshaped into matrix forms, then the leading order dispersion around $\pmb k$ is linear. Therefore, we need a projection operator $P_H$ to project the bases carrying the dual vector Reps onto the hermitian subspace (i.e. the union of real symmetric subspace and the imaginary anti-symmetric subspace). Performing the projection $P_H$ and reshaping the remaining bases into hermitian matrices, then we obtain the explicit form of $\Gamma^{1,2,3}$,
\Beq
\Gamma^m = \sum_{i=1}^p r_i \gamma^m_i,
\Eeq
where $p$ is the mutiplicity of the dual vector Rep(s) contained in the product Rep, $r_i\in\mathbb R$ are arbitrary real numbers, and $\gamma_i^{1,2,3}, i=1, ... ,p$ are the bases of the $i$th dual vector Rep. Substituting into (\ref{linear}) and (\ref{kdotp_ham_deltak}) we obtain the $k\cdot p$ effective model. 

In next section, we will introduce an alternative method to obtain the hermitian matrices $\Gamma^{1,2,3}$ without using the projection operator $P_H$.

\section*{The general case }\label{kdotp:General cases}

Generally, $T_0$ acts on $\delta\pmb k$ in the following way,
\Beq
T_0 \delta k_m=\sum_n D^{(\bar v)}_{mn}(T_0) \delta k_n,
\Eeq
where $D^{(\bar v)}(T_0)$ is a $3\times 3$ real matrix, and $D^{(\bar v)}(T_0)K$ can be considered as part of the dual vector Rep of the anti-unitary group $G_0(\pmb k)$. 
Accordingly, $\Gamma^m$ should vary in the following way in anology to (\ref{kdotp_Gampro_F2}),
\begin{align}\label{kdotp:GeneralGamma}
M(T_0) (\Gamma^{m})^* M^\dag (T_0)= \sum_n \Gamma^n D^{(\bar v)}_{nm}(T_0).
\end{align}

Similar to the discussion in Sec.\ref{TypeI}, we introduce a bases transformation $\tilde \psi_{\pmb k}=M^*(T_0)\psi_{\pmb k}$, then
\Beq
\hat h\tilde \psi_{\pmb k}\hat h^{-1} = F^T(h)\tilde \psi_{\pmb k},\ \  \hat T_0 \tilde \psi_{\pmb k}\hat T_0^{-1} = M^T(T_0) K\tilde \psi_{\pmb k},
\Eeq
where $F(h)$ is defined in (\ref{Fh}). We further define 
\beq\label{Gammaprime}
\tilde {\Gamma}^m =  \Gamma^m  M^T(T_0), 
\eeq
then (\ref{kdotp_Gampro_F1}) and (\ref{kdotp:GeneralGamma}) deform into 
\beq\label{tildeGaH}
M(h)\tilde{\Gamma}^{m}F^T(h) &=& \sum_n D^{(\bar v)}_{nm}(h) \tilde{\Gamma}^{n},\label{tildeGaH}\\
M(T_0) \big(\tilde\Gamma^m\big)^*M^T(T_0)&=&\sum_nD^{(\bar v)}_{nm}(T_0)\tilde\Gamma^n,\label{tildeGaT} 
\eeq
respectively.

Considering the set of matrices $\tilde\Gamma^{1,2,3}$ as a single column vector $\tilde \Gamma$ with 
\beq\label{Gamma3dd}
(\tilde \Gamma)_{n\times d^2+i\times d+j} = (\tilde\Gamma^n)_{ij},
\eeq 
then $\tilde\Gamma$ carries the identity Rep of $G$,
\beq\label{WGm}
W(h)\tilde \Gamma=\tilde \Gamma,\ \  W(T_0)K \tilde \Gamma=\tilde \Gamma K,
\eeq
where 
\Beq
&&\!\! W(h)=D^{(v)}(h)\otimes V(h)=D^{(v)}(h)\otimes M(h)\otimes F(h),\\
&&\!\!W(T_0)\!=\!D^{(v)}(T_0)\! \otimes\! V(T_0)\!=\!D^{(v)}(T_0)\!\!\otimes\!\! M(T_0)\!\otimes\! M(T_0).
\Eeq 
Here the relation $\sum_{m} D^{(\bar v)}_{mi}(h)D^{(v)}_{mj}(h) =\delta_{ij}$ has been used. Later we will alternately use the notation $\tilde\Gamma^m$ and $\tilde\Gamma$.

Taking transpose of (\ref{Gammaprime}), we have
\beq\label{kdotp_tildeGamma}
(\tilde {\Gamma}^m)^T&=&M(T_0) (\Gamma^m)^T \notag \\
&=& M(T_0)\sum_n D^{(v)}_{mn}(T_0) M^*(T_0) \Gamma^n M^T(T_0) \notag\\
&=&\eta_0 M(\sigma) \sum_n  D^{( v)}_{mn} (T_0) \tilde \Gamma^n,
\eeq
where we have used the transpose of (\ref{kdotp:GeneralGamma}) namely $\sum_n (\Gamma^n)^T D^{(\bar v)}_{nm}(T_0) = M^*(T_0) \Gamma^m M^T(T_0)$ and the hermitian condition $\Gamma^m=(\Gamma^m)^\dag$. Namely, the symmetry condition in (\ref{kdotp_tildeGamma}) is a consequence of the requirements that $\tilde {\Gamma}^m$  should be hermitian and $T_0$-symmetric. 

Taking the transpose of $\tilde \Gamma$ twice, we obtain $\big((\tilde {\Gamma}^m)^T\big)^T=\eta_0^2\sum_l D_{ml}^{(v)}(\sigma) M(\sigma)\tilde \Gamma^lM^T(\sigma)=\sum_l D_{ml}^{(v)}(\sigma) M(\sigma)\tilde \Gamma^lF^T(\sigma)$. Therefore, the self-consistency requires that $\sum_l D_{ml}^{(v)}(\sigma) M(\sigma)\tilde \Gamma^lF^T(\sigma)=\tilde {\Gamma}^m$, which is ensured by (\ref{tildeGaH}).

Noticing $\tilde\Gamma^m=[\eta_0 M(\sigma) \sum_n  D^{( v)}_{mn} (T_0) \tilde \Gamma^n]^T$, we can define the generalized twist operator which transforms $\tilde\Gamma^m$ to $[\eta_0M(\sigma) \sum_n  D^{( v)}_{mn} (T_0) \tilde \Gamma^n]^T$,
\Beq
\mathbb T_{\eta_0} &=& \eta_0 (I\otimes \mathscr T)\Big(D^{( v)}(T_0)\otimes M(\sigma)\otimes I\Big)\\
&=& D^{( v)}(T_0)\otimes \mathscr T_{\eta_0}, 
\Eeq
where $\mathscr T_{\eta_0}$ is defined in (\ref{scrTeta0}). Thus $\mathbb T_{\eta_0}$ defines a generalized `transpose' of $\tilde\Gamma$, and the self-consistency condition is $(\mathbb T_{\eta_0})^2 \tilde \Gamma =W(\sigma)\tilde\Gamma=\tilde\Gamma$. In the eigenspace of $W(\sigma)$ with eigenvalue 1, the operator $\mathbb T_{\eta_0}$ has eigenvalues $\pm1$. If we call a vector satisfying 
\beq\label{Teta0}
\mathbb T_{\eta_0} \tilde\Gamma =\tilde\Gamma
\eeq
to be $\eta_0$-symmetric, then the projector onto the $\eta_0$-symmetric subspace is
\[
\mathbb P_{\eta_0} = {1\over2} \mathbb P_\sigma(I+  \mathbb T_{\eta_0}),
\]
where $\mathbb P_\sigma$ is the projector onto the eigenspace of $W(\sigma)$ with eigenvalue 1. 

Denoting $\mathbb P^{(I)}$ as the projector onto the subspace of identity Reps $$\mathbb P^{(I)}={1\over|H|}\sum_{h\in H}W(h),$$ then obviously $\mathbb P^{(I)}\mathbb P_\sigma=\mathbb P^{(I)}$. The condition for the linear dispersion is that the $\eta_0$-symmetric subspace contains the identity Rep of $H$, namely
\[
a_{(I)}^H = {\rm Tr}\big( \mathbb P^{(I)} \mathbb P_{\eta_0} \big)  = {1\over2} {\rm Tr}\big( \mathbb P^{(I)}  (I+  \mathbb T_{\eta_0})\big)\geq1.
\]

Defining the matrices  
\[
W_{\eta_0}(h)= {1\over2}W(h)  (I+  \mathbb T_{\eta_0}),
\]
with the matrix entries
\Beq
[W_{\eta_0}(h)]_{mkl,nij}\!\! &=&\!\! {1\over2}\Big[ D_{mn}^{(v)}(h) M_{ki}(h) F_{lj}(h) \\
&+&\!\! \eta_0 \sum_a D_{mn}^{(v)}(hT_0) M_{kj}(h) F_{la}(h)M_{ai}(\sigma) \Big],
\Eeq
then the condition of linear dispersion reduces to $a_{(I)}^H = {\rm Tr\ } (\mathbb P^{(I)}\mathbb P_{\eta_0}) ={1\over 2|H|}\sum_{h\in H} {\rm Tr\ }W_{\eta_0}(h)\geq1$. From the above expression of $W_{\eta_0}(h)$, after some calculations (see Appendix \ref{app:proof}) we obtain 
{\footnotesize
\begin{empheq}[box=\fbox]{align}\label{general_form_final}
a_{(I)}^H\!\!=\!\!{1\over 2|H|}\! \sum_h \!\!\Bigg[ \! |\chi(h)|^2\chi^{(v)}\!(h)
\!+\! {\chi^{(v)}\!(hT_0)} \omega(hT_0, hT_0) \chi((hT_0)^2)\!\Bigg]\!.
\end{empheq}
}
Specially, if $D^{(v)}(T_0)=I$, then ${\chi^{(v)}(hT_0)}=\chi^{(v)}(h)$, above formula reduces to the equation (\ref{kdotp_judge_Final});  if $D^{(v)}(T_0)=-I$, then ${\chi^{(v)}(hT_0)} = -\chi^{(v)}(h)$, above formula can be simplified as
\Beq
a_{(\bar v)}^H\!=\!{1\over2 |H|}\! \!\sum_h \!\left[ \chi(h)\chi^*(h) \! -\!\omega(hT_0, hT_0) \chi((hT_0)^2)\right]\!\chi^{(v)}\! (h).
\Eeq

When the vector Rep $(v)$ of $G$ is reducible, then the dispersions may be different along different directions.  In this case, we need to reduce the vector Rep $(v)$ and check the resultant irreducible Reps one by one. For instance,  if $H=D_{4h}$, then the vector Rep is reduced to $(v)=E_{u}+A_{1u}$, where $(k_x, k_y)^T$ vary in the rule of the Rep $E_u$  and $(k_z)$ vary in the rule of $A_{1u}$. In this case, we need to replace $(v)$ in (\ref{kdotp_judge_Final}) or (\ref{general_form_final}) by $(E_u)$ and $(A_{1u})$. If $a_{(E_u)}^H$ is nonzero, then the dispersion along $k_x, k_y$ is linear, otherwise the dispersion is quadratic or of higher order. Similarly, if $a_{(A_{1u})}^H\neq 0$, then the dispersion along $k_z$ is linear.

\subsubsection{Procedure of obtaining the $\Gamma^{1,2,3}$ matrices}

As $a_I^H\geq1$, above procedure provides another way to obtain the matrices $\Gamma^{1,2,3}$ besides the method of reducing the product Rep $M(g)\otimes M^*(g)K_{s(g)}, g\in G$ into direct sum of IPReps.  We only consider the case where the 3-dimensional vector Rep $(v)$ is irreducible. The procedure is easily generalized to the cases in which $(v)$ is reducible.

Firstly, obtain the eigenspace of $\mathbb P^{(I)}\mathbb P_{\eta_0}={1\over2}\big( \mathbb P^{(I)}  (I+  \mathbb T_{\eta_0})\big)$ with eigenvalue 1. Supposing the dimension of this eigenspace is $p\geq1$,  choose a set of orthonormal bases $ \tilde\zeta_1, \tilde\zeta_2, ..., \tilde\zeta_p$. 

Secondly,  tune the bases in above subspace such that $T_0$ is represented as $IK$. To this end, calculate the Rep of $T_0$
\[
\mathscr M_{ij}(T_0) K = \tilde \zeta_i^\dag W(T_0) K \tilde \zeta_j,
\]
and construct the new bases
\[
\tilde \Delta_j=\sum_i \tilde \zeta_i \mathscr M^{1\over2}_{ij}(T_0).
\]
Then each of these new bases carries the identity Reps of the total group $G$ (see Appendix~\ref{app:condition} for details).

Thirdly, from (\ref{Gamma3dd}), we can decouple each eigenvector $\tilde \Delta_i$ as
$$
\tilde\Delta_i ={1\over\sqrt{3}}e_1\otimes \tilde \gamma_i^1+{1\over\sqrt{3}}e_2\otimes \tilde \gamma_i^2+{1\over\sqrt{3}}e_3\otimes \tilde \gamma_i^3, 
$$  
where $e_1=(1,0,0)^T, e_2=(0,1,0)^T,e_3=(0,0,1)^T$ stand for the bases  $\pmb a_1,\pmb a_2,\pmb a_3$ (not necessarily orthogonal) of the vector Rep $(v)$ respectively, and $\tilde \gamma_i^m$ is the Schmidt partner of $e_m$. It can be shown that (see Appendix~\ref{app:condition}), $\tilde \gamma_i^m$ satisfies the relations (\ref{tildeGaH}) and (\ref{tildeGaT}), and that $\gamma_i^m=\tilde \gamma_i^mM^*(T_0)$ is a hermitian matrices for any $i=1,...p$ and $m=1,2,3$.

The general form of the vector $\tilde \Gamma$ is a linear combination of the bases $\tilde \Delta_i$,
\beq\label{Schmidt}
\tilde \Gamma &=&r_1 \tilde \Delta_1 + r_2 \tilde\Delta_2 +...+r_p \tilde\Delta_p \notag\\
&=&{1\over\sqrt{3}} e_1 \otimes\tilde\Gamma^1+{1\over\sqrt{3}} e_2 \otimes\tilde\Gamma^2 +{1\over\sqrt{3}} e_3 \otimes\tilde\Gamma^3,
\eeq
where $r_1,...,r_p\in \mathbb R$ are non-universal real constants.
From (\ref{Schmidt}) and the relation $\Gamma^m =   \tilde \Gamma^m M^*(T_0)$, we obtain the matrices $\Gamma^m$,
\Beq
\Gamma^m =    \sum_{i=1}^p r_i \Big(\tilde\gamma_i^m M^*(T_0)\Big) =\sum_{i=1}^p r_i \gamma_i^m.
\Eeq

\subsubsection{Higher Order Dispersions and Nodal Lines}

The discussion of linear dispersion can be straightforwardly generalized to higher order dispersions. Suppose that a set of order-$N$ homogeneous polynomials $$
P^{(N)}_i(\delta \pmb k)=\sum_{a+b+c=N} f^{(N)}_{i(abc)}\delta k_1^a\delta k_2^b\delta k_3^c
$$ 
carry a linear Rep $(\bar \mu)$ of the group $G$, the existence of the dispersion
\[
H_{\pmb k+ \delta\pmb k} = \sum_{i}  P^{(N)}_{i} (\delta \pmb k) \psi_{\pmb k+\delta\pmb k}^\dag\Gamma^{(N)}_i \psi_{\pmb k+\delta\pmb k}
\]
can be judged using the formula (\ref{general_form_final}) with the vector Rep $(v)$ replaced by the linear Rep $(\mu)$ (see Appendix \ref{app:vec} for an example). The method of obtaining the corresponding matrices $\Gamma^{(N)}_i$ is also similar. 

If the vector Rep of $G$ is reducible, it is possible that the degeneracy is lifted along some directions (such as the $k_x, k_y$ directions) but are preserved along certain direction (such as the $k_z$ direction) to form a nodal line. The little co-group on the line is generally smaller than the one on the conner of the BZ. If the IPRep of the little co-group at the conner of the BZ is still irreducible along a certain line, then this line is a nodal line. Therefore, the existence of the nodal line can be judge from the formula (\ref{IR_condition}) \cite{YFL}. The same method can be applied to judge the stability of the degeneracy under external perturbations (see section \ref{sec:response}).

\subsection{Response to External Probe Fields}\label{sec:response}

The IPRep $M(g)K_{s(g)}, g\in G$ of anti-unitary symmetry group $G$ results in energy degeneracy in single-particle spectrum. Here we discuss the possible lifting of the degeneracy under external probe fields, such as $\pmb E$ and $\pmb B$, stain, or temperature gradience, etc. We assume that the probe fields carry irreducible linear Reps of the group $G$. For instance, electric fields $\pmb E$ or magnetic fields $\pmb B$ carry vector Reps of the unitary subgroup $H$, but they vary differently under the anti-unitary element $T_0$ since $\pmb E$ is invariant under time reversal while $\pmb B$ reverse its sign under time reversal.

There are two possible consequences under external probes. The first possible result is that the degeneracy guaranteed by the IPRep $M(g)K_{s(g)}$ is preserved. The other possibility is that the degeneracy is lifted in linear or higher order terms of the probe fields.

To judge if the probe fields can lift the degeneracy or not (summing over all orders of perturbation), we need to know the remaining symmetry group with the presence of the perturbation, and then judge if the restrict Rep is reducible or not.  Suppose the probe field reduces the symmetry group from $G=H+T_0H$ to $G'=H'+T'_0H'$ where $T_0'$ is anti-unitary. If the irreducible Rep $M(g)K_{s(g)}$ of $G$ remains irreducible for $G'$, namely if
\begin{align}
	{1\over |H'|} \sum_{h\in H'} {1\over2}\left[ \chi(h)\chi^*(h)  +\omega_2(T'_0h, T'_0h) \chi((T'_0h)^2)\right] = 1. \notag
\end{align}
holds for the group $G'$, then the degeneracy is robust against this perturbation.

If the left hand side of above equation is not equal to 1, then the restricted Rep is reducible and the degeneracy can be lifted at certain order. In the following we only discuss the linear splitting by external fields, such as $\pmb E$ and $\pmb B$. The linear response is given by the perturbed Hamiltonian in form of
\begin{align}
	H=\sum_{\pmb k} \psi^\dag_{\pmb k}(\pmb E\cdot \pmb P +  \pmb B\cdot \pmb M)\psi_{\pmb k}
\end{align}
where $P^m, M^m$ are CG matrices similar to the $\Gamma^m$ matrices discussed before. The existence of linear coupling terms $\pmb P$ (or $\pmb M$) can be checked using the criterion  (\ref{general_form_final}) with $M^{(v)}(T_0)$ the transformation matrix of $\pmb E$ (or $\pmb B$) under the action of $T_0$.

\section{Conclusions and Discussions}\label{sec:sum} 

In summary, from a physical approach, we derived the condition (\ref{IR_condition}) for the irreducible projective representations of anti-unitary groups. This approach provides a practical method to reduce an arbitrary projective Rep into a direct sum of irreducible ones, which is applicable for either unitary or anti-unitary groups. 

As a physical application of this approach, for single particle systems with magnetic space group symmetry, we provide the method to construct the $ k\cdot p$ perturbation theory at the high symmetry point of the Brillouin zone. We provide the criterion (\ref{general_form_final}) to judge if the dispersion is linear or of higher order, and then provide the method to obtain the corresponding $ k\cdot p$ Hamiltonian up to a few non-universal constants. 

In the present work, we assume that the quasiparticles vary under linear representations of the magnetic space groups. However, in strongly interacting systems, projective representations of the magnetic space groups can emerge in the fractionalized low-energy quasiparticle excitations for systems with intrinsic topological order. We leave the discussion of this situation for future study. 

{\it Acknowledgements} We thank L. J. Zou and Y. X. Zhao for helpful discussions. Z.Y.Y and Z.X.L. are supported by the Ministry of Science and Technology of China (Grant No. 2016YFA0300504), the NSF of China (Grants No.11574392 and No. 11974421), and the Fundamental Research Funds for the Central Universities and the Research Funds of Renmin University of China (Grant No. 19XNLG11). J. Yang and C. Fang are supported by Ministry of Science and Technology of China under grant number 2016YFA0302400, National Science Foundation of China under grant number 11674370 and Chinese Academy of Sciences under grant number XXH13506-202 and XDB33000000.


\appendix

\section{ Hermiticity of $\Gamma$ and validity of (\ref{Irrep-I}) for general anti-unitary groups}\label{app:Peta0} 

In this appendix, we firstly prove a lemma, then introduce a theorem. Type-I and type-II anti-unitary groups are treated on equal footing.\\
\begin{lemma}\label{lemma2}
If the unitary representation $M(T_0)K$ of an anti-unitary element $T_0$ satisfies the condition $ [M(T_0)K]^2=M(T_0)M^*(T_0) =I$, then there exist a unitary matrix $U$ such that $U^\dag M(T_0)KU = U^\dag M(T_0)U^*K =IK$.
\end{lemma}
\begin{proof}
Since $[M(T_0)]^\dag = [M(T_0)]^{-1} =[M(T_0)]^*$, letting $U = [M(T_0)]^{1\over2}$, we have $U^\dag = U^* = [M(T_0)]^{-{1\over2}}$. Consequently, one can easily verify that
\Beq
U^\dag M(T_0)K U &= &U^\dag M(T_0)U^* K\\
&=&[M(T_0)]^{-{1\over2}}  M(T_0)  [M(T_0)]^{-{1\over2}} K\\
&=&IK.
\Eeq
\end{proof}

\begin{corollary}
In a linear Rep of anti-unitary group $G$, one can chose bases in the eigenspace of $\sigma \equiv T_0^2$ with eigenvalue 1 such that $T_0$ is represented as $IK$ in this subspace.
\end{corollary}

\begin{proof}
Suppose $D(g)K_{s(g)}, g\in G$ is a $N$-dimensional linear Rep of $G$, $\Delta_1,...\Delta_N$ are the orthonormal bases of the eigenspace of $\sigma=T_0^2$ with eigenvalue 1, namely.
\[
D(\sigma)\Delta_i =\Delta_i.
\]
Since $\sigma T_0=T_0\sigma$, and accordingly $D(\sigma)D(T_0)K=D(T_0)KD(\sigma)$, we have
\Beq
D(\sigma)D(T_0)K\Delta_i=D(T_0)K \big(D(\sigma)\Delta_i\big) = D(T_0)K \Delta_i,
\Eeq
namely, the eigenspace of $\sigma$ is closed under the action of $T_0$. In the eigenspace of $\sigma$ with eigen value 1, the Rep of $T_0$ takes the form
\Beq
\mathscr M(T_0)_{ij}K = \tilde \Delta_i^\dag D(T_0)K \tilde \Delta_j,
\Eeq
with $\mathscr M^\dag(T_0)=\mathscr M^{-1}(T_0)$ and $\mathscr M(T_0)\mathscr M^*(T_0)=\mathscr M(\sigma)=I$. From lemma \ref{lemma2}, the Rep of $T_0$ can be transformed into $IK$ in the new bases,
\[
\Delta'_i = \sum_{j=1}^N \Delta_j \mathscr M_{ji}^{1\over2}(T_0),
\]
with 
\[
\mathscr M(T_0) K \Delta'_i = \mathscr M(T_0) \Delta^{'*}_i K= \Delta^{'*}_i K.
\]
\end{proof}

Now it is ready to introduce the theorem.\\
\begin{theorem}\label{Theorem1}
If $\tilde \Gamma$ is a common eigenvector $\tilde \Gamma$ of $V_{\eta_0}(h), h\in H$ with eigenvalue 1, namely, $V_{\eta_0}(h)\tilde\Gamma=\tilde\Gamma$ for all $h\in H$, with $V_{\eta_0}(h)=\frac{1}{2} [V(h) (I + \mathscr T_{\eta_0})]$ and $V(h)=M(h)\otimes F(h)$, then $\tilde\Gamma$ has the following properties: \\
1) it carries the identity Rep of $H$; \\
2) it is $\eta_0$-symmetric; \\
3) if the basis satisfies $V(T_0)\tilde\Gamma^*=\tilde\Gamma$ ({\it i.e.} if $\tilde\Gamma$ carries the identity Rep of $G$), then $\Gamma=\tilde\Gamma M^*(T_0)$ is an hermitian matrix where $\tilde\Gamma$ has been reshaped into a matrix.  
\end{theorem}
\begin{proof}
Firstly, since $V(h)=M(h)\otimes F(h)$ is a linear Rep of $H$, we define the following projection operator 
$$P^{(I)}={1\over |H|}\sum_{h\in H} V(h),$$
which projects from the product space $V(h)=M(h)\otimes F(h)$ onto the identity Rep space. Accordingly,  we have
$${1\over |H|}\sum_{h\in H} V_{\eta_0}(h)=P^{(I)}P_{\eta_0}.$$

Supposing $\tilde \Gamma$ is a common eigenvector of $V_{\eta_0}(h), h\in H$ with eigenvalue 1,  then we have
\beq\label{Veta0Gamma}
{1\over |H|}\sum_{h\in H} V_{\eta_0}(h) \tilde\Gamma&=&P^{(I)}P_{\eta_0} \tilde \Gamma=\tilde\Gamma.
\eeq
Therefore, 
\[
P^{(I)} \tilde\Gamma = P^{(I)} P^{(I)}P_{\eta_0}\tilde\Gamma=P^{(I)}P_{\eta_0}\tilde\Gamma=\tilde\Gamma,
\]
namely, $\tilde\Gamma$ is the CG coefficient coupling the direct product Rep $V(h) = M(h)\otimes F(h)$ to the identity Rep. 

On the other hand, from $P^{(I)}\tilde\Gamma=\tilde\Gamma$ and $P^{(I)} P_{\eta_0}\tilde\Gamma=\frac{1}{2}P^{(I)} (I+\mathscr T_{\eta_0})\tilde\Gamma =\tilde\Gamma$,
we have 
\[
 {1\over2} (I+\mathscr T_{\eta_0})\tilde\Gamma =\tilde\Gamma.
\]
Therefore $\mathscr T_{\eta_0}\tilde\Gamma =\tilde\Gamma $. By definition, $\mathscr T_{\eta_0}\tilde\Gamma = [\eta_0M(\sigma)\tilde\Gamma]^T$, so we have $$\tilde\Gamma^T=\eta_0 M(\sigma)\tilde\Gamma.$$ Namely, $\tilde\Gamma$ is $\eta_0$-symmetric. Especially, for type-I anti-unitary groups, $\sigma=E$, the $\eta_0$-symmetry reduces to $\tilde\Gamma^T=\eta_0\tilde\Gamma$.

Thus we have verified that the eigenvector $\tilde\Gamma$ of $P^{(I)} P_{\eta_0}\tilde\Gamma=\tilde\Gamma$ indeed has the properties 1) and 2).

Now we illustrate that the relation $V(T_0)K \tilde\Gamma = \tilde\Gamma$ can be satisfied.

Firstly we shows that $V(T_0)K$ preserves the eigenspace of $P^{(I)}$.  Noticing that $T_0(\sum_{h\in H} h) = (\sum_{h\in H} h) T_0$, and that $V(g)K_{s(g)}, g\in G$ is a linear Rep of $G$, so $V(T_0)K P^{(I)}= P^{(I)}V(T_0)K$. Therefore, if $\tilde\Gamma$ is an eigenvector of $P^{(I)}$ with $P^{(I)}\tilde\Gamma=\tilde\Gamma$, then
\[
P^{(I)} V(T_0)K \tilde\Gamma = V(T_0)K P^{(I)} \tilde\Gamma = V(T_0)K\tilde\Gamma.
\]
Namely, $V(T_0)K\tilde\Gamma$ is still an eigenvector of $P^{(I)}$.

Then we show that $V(T_0)K$ also preserves the eigenspace of $P_{\eta_0}$. From the definition of the unit twist operator $[(\mathscr T)_{kl,ij}=\delta_{kj}\delta_{li}]$, it is easily to verify that 
$$ (X\mathscr T)_{kl,ij}=X_{kl,ji}$$ and $$(\mathscr T X)_{kl,ij}=X_{lk,ij}$$ for arbitrary matrix $X$. Similarly, for a direct product matrix $(X\otimes Y)_{kl,ij}=X_{ki}Y_{lj}$, the twist operator acts as
\begin{align}
\begin{cases}
[(X\otimes Y)\mathscr T]_{kl,ij}=(X\otimes Y)_{kl,ji}=X_{kj}Y_{li}\\
[\mathscr T(X\otimes Y)]_{kl,ij}=(X\otimes Y)_{lk,ij}=X_{li}Y_{kj}
\end{cases}\notag,
\end{align}
which gives 
\begin{align}
[\mathscr T(X\!\otimes\! Y)\mathscr T]_{kl,ij}&=[(X\!\otimes \!Y)\mathscr T]_{lk,ij}\notag\\
&=(X\!\otimes\! Y)_{lk,ji}=X_{lj}Y_{ki}\!=\!(Y\!\otimes\! X)_{kl,ij}. \notag
\end{align}

Since $V(T_0)=M(T_0)\otimes M(T_0)$, we have 
\begin{align}
&[\mathscr T_{\eta_0} V(T_0)K \mathscr T_{\eta_0}]\notag \\ 
=&\eta_0\eta_0^*\mathscr T[M(\sigma)\otimes I][M(T_0)\otimes M(T_0)]\mathscr T [M^*(\sigma)\otimes I]  K \notag\\ %
=&[I\otimes M(\sigma)][M(T_0)\otimes M(T_0)][M^*(\sigma)\otimes I]  K\notag\\
=& [M(T_0)\otimes M(T_0)][M(\sigma)\otimes F(\sigma)]^* K\notag \\
=&V(T_0)K(\mathscr T_{\eta_0})^2, 
\end{align}
where we have used $(\mathscr T_{\eta_0})^2=V(\sigma)$. Therefore, if $\tilde\Gamma$ is $\eta_0$-symmetric, $i.e.$, if $\ \mathscr T_{\eta_0}\tilde\Gamma =\tilde\Gamma$,   then $\mathscr T_{\eta_0} V(T_0)K \tilde\Gamma = \mathscr T_{\eta_0} V(T_0)K \mathscr T_{\eta_0}\tilde\Gamma=V(T_0)K(\mathscr T_{\eta_0})^2\tilde\Gamma=V(T_0)K\tilde\Gamma$,
which means that
$\big(V(T_0)K\big) \tilde\Gamma$ is still $\eta_0$-symmetric.

Therefore, the eigenspace of $P^{(I)}P_{\eta_0}$ with eigenvalue 1  
is preserved under the action of $T_0$. Namely, this eigspace form a linear Rep of the anti-unitary group $G$. Noticing that $V(\sigma)=V(T_0)V^*(T_0)$ is represented as an identity matrix in this eigenspace, from lemma \ref{lemma2}, we can `diagonalize' $V(T_0)K$ as $IK$ in this subspace. Namely, we can choose proper bases such that $M(T_0) \otimes M(T_0)\tilde\Gamma^*=\tilde\Gamma$, or equivalently (\ref{T0tG}), holds. 

From the transpose of (\ref{T0tG}), we have $M(T)\tilde\Gamma^\dag [M(T)]^T=\tilde\Gamma^T=\eta_0M(\sigma)\tilde\Gamma$, where the symmetry equation has been used. Substituting $\Gamma=\tilde \Gamma M^*(T_0)$ into above equation and noticing $M(T_0)M^*(T_0)=\eta_0M(\sigma)$, we finally obtain the hermitian condition $\Gamma^\dag =\Gamma$.  Thus the property 3) has been verified.\\
\end{proof}

Noticing that the eigenspace of $P^{(I)}$ is either $\eta_0$-symmetric or $(-\eta_0)$-symmetric, so $ P^{(I)}P_{\eta_0}$ is also a projection operator 
\[
\big( P^{(I)}P_{\eta_0}\big)^2= P^{(I)}P_{\eta_0},
\]
therefore its eigenvalues are either 1 or 0. If $M(g)K_{s(g)}, g\in G$ is irreducible, then $P^{(I)}P_{\eta_0}$ has only one nonzero eigenvalue. The trace of this projection operator ${\rm Tr}\ (P^{(I)}P_{\eta_0})=1$ yields the irreduciblity condition (\ref{Irrep-I}).

\section{Derivation of (\ref{Anti}) for Type-II anti-unitary groups} \label{app:proof}

Following the same discussion of type-I anti-unitary groups, for type-II anti-unitary groups we obtain the matrix form $V_{\eta_0}(h)=\frac{1}{2} [M(h)\otimes F(h) (I + \mathscr T_{\eta_0})]$, namely
\begin{align}
[V_{\eta_0}(h)]_{kl,ij}=\frac{1}{2} \Big(&\!\! \left[M(h)\!\otimes\! F(h) \right]_{kl,ij}  \notag\\
	& +  \eta_{0}[M(h)\! \otimes\! F(h)M(\sigma)]_{kl,ji}\!\! \Big).
\end{align}
From theorem \ref{Theorem1} in appendix \ref{app:Peta0}, we can start with the equation (\ref{typeII:criterion}),
which can be expressed in terms of characters as
 \beq\label{trace_Mnew}
	1\!\!=\!\!\frac{1}{2|H|}\!\sum_{h,i,j}\!\Big(\! M_{ii}(h) F_{jj}(h) \!+\!\eta_{0} M_{ij}(h) [F(h)M(\sigma)]_{ji}\! \Big).
\eeq
 Remembering that $M(\sigma)=\eta^{-1}_{0}M(T_0)M^*(T_0)$, the second term in (\ref{trace_Mnew}) can be transformed into
\begin{align}\label{latter_term}
	\sum_{i,j}&\eta_{0} M_{ij}(h) [F(h)M(\sigma)]_{ji} =\eta_{0}{\rm Tr}[M(h)F(h)M(\sigma)]\notag \\
	=& \eta_{0}{\rm Tr}[M(h)F(h)\eta^{-1}_{0}M(T_0)M^*(T_0)]\notag \\
	=&{\rm Tr}[M(h)M(T_0)M^*(h)M^\dag(T_0)M(T_0)M^*(T_0)]\notag \\
	=&{\rm Tr}[M(h)M(T_0)M^*(h)M^*(T_0)]\notag \\
	=&\omega(h,T_0)\omega^*(h,T_0){\rm Tr}[M(hT_0)M^*(hT_0)]\notag\\
	=&\omega(h,T_0)\omega^*(h,T_0) \omega(hT_0,hT_0)\chi((hT_0)^2).
\end{align}
Finally, noticing ${\rm Tr\ } M(h)=\chi(h)$ and ${\rm Tr\ } F(h)=\chi^*(h)$, (\ref{trace_Mnew}) reduces to (\ref{Anti}), namely, 
\begin{align}
	1=\frac{1}{2|H|}\sum_{h\in H}\left[ \chi(h)\chi^*(h)  +\omega(hT_0, hT_0) \chi((hT_0)^2)\right].
\end{align}

\section{$k\cdot p$ theory: Derivation of (\ref{kdotp_Gampro_F1}), (\ref{kdotp_Gampro_F2}), (\ref{kdotp:GeneralGamma}) and Discussion for General Dispersions}\label{app:vec}
We starts with the equation (\ref{Gammak}), namely,
$$M(g)K_{s(g)} \Gamma(g^{-1}\delta \pmb k) K_{s(g)} M^\dag (g) = \Gamma(\delta\pmb  k).$$

Letting $\delta \pmb k'=g^{-1}\delta \pmb k$, then $\delta \pmb k=g\delta \pmb k'$ and (\ref{Gammak}) becomes
$$M(g)K_{s(g)} \Gamma(\delta \pmb k') K_{s(g)} M^\dag (g) = \Gamma(g \delta\pmb  k').$$
Since the summation over $\delta\pmb k'$ is equivalent to the summation over $\delta \pmb k$, therefore we have
\beq\label{B1}
M(g)K_{s(g)} \Gamma(\delta \pmb k) K_{s(g)} M^\dag (g) = \Gamma(g \delta\pmb  k).
\eeq
If there is a linear dispersion then $\Gamma(\delta\pmb k)=\sum_{m=1}^3\delta k_m\Gamma^m$. Notice that $\delta\pmb k$ varies as dual vector under the action of the unitary subgroup $H$, namely $\hat h\delta k_m = \sum_n D^{(\bar v)}_{mn}(h)\delta k_n$. Substituting these relations into (\ref{B1}) and letting $g=h\in H$, then we have 
\begin{align}
M(h)\left( \sum_{n} \Gamma^n\delta k_n \right) M^{\dag}(h)= \sum_{m,n}\Gamma^m D^{(\bar v)}_{mn}(h)\delta k_n. 
\end{align}
Thus the equation (\ref{kdotp_Gampro_F1}) is proved, $i.e.$, $ M(h)\Gamma^n M(h)^\dag=\sum_m D^{(\bar v)}_{mn}(h)\Gamma^m$.

Now consider the anti-unitary element $g=T_0$. From (\ref{B1}), we obtain 
$$
M(T_0 )K\Gamma(\delta \pmb k) K M^\dag (T_0 ) = \Gamma(T_0\delta\pmb  k). 
$$
If $T_0$ has a nontrivial action on $\delta \pmb k$, namely $T_0 \delta k_m=\sum_n D^{(\bar v)}_{mn}(T_0) \delta k_n$, then linear dispersion $\Gamma(\delta\pmb k)=\sum_{m=1}^3\delta k_m\Gamma^m$ indicates that
\Beq
M(T_0 )K \left(\sum_n \Gamma^n\delta k_n\right) K M^\dag (T_0 ) = \sum_{mn} D^{(\bar v)}_{mn}(T_0)\Gamma^m\delta k_n, 
\Eeq
which is equivalent to  (\ref{kdotp:GeneralGamma}), $i.e.$, $M(T_0) (\Gamma^{n})^* M^\dag (T_0)= \sum_m \Gamma^m D^{(\bar v)}_{mn}(T_0).$ Here we have used the fact that $\delta k_n\in \mathbb{R}$ are real numbers. (\ref{kdotp_Gampro_F2}) is a special case of (\ref{kdotp:GeneralGamma}) with $D^{(\bar v)}_{mn}(T_0)=I$.

Similar discussion can be generalized to the case when the vector Rep is reducible, or to the cases where the dispersions are of higher order. Generally, the object $\sum_n\Gamma^n\delta k_n$ can be replaced by $\sum_{i}\Gamma^{(N)}_i P_i^{(N)}(\delta\pmb k)$, where 
\[
P^{(N)}_i(\delta\pmb k) = \sum_{a+b+c=N} f^{(N)}_{i(abc)}\delta k_1^a\delta k_2^b\delta k_3^c,\ \ f_{abc}\in\mathbb R
\]
belongs to a set of order-$N$ homogeneous polynomials of $\delta k_{1},\delta k_2,\delta k_3$ which vary under the rule of irreducible linear Rep of $G$. 

For instance, in the case $H=\mathscr C_{6v}$, the quadratic polynomials $(P^{(2)}_1, P^{(2)}_2)^T=(k_x^2-k_y^2, 2k_xk_y)^T$ vary as a two-component column vector under the irreducible Rep $(E_2)=(\bar E_2)$, namely, 
\Beq
&&h P^{(2)}_i(\delta \pmb k) = \sum_j D^{(\bar E_2)}_{ij}(h) P^{(2)}_j(\delta \pmb k),\\
\Eeq
for $h\in \mathscr C_{6v}$, and 
\Beq
&&T_0 P^{(2)}_i(\delta \pmb k) = \sum_j D^{(\bar E_2)}_{ij}(T_0) P^{(2)}_j(\delta \pmb k).
\Eeq
Accordingly,  similar to (\ref{kdotp_Gampro_F1}) and (\ref{kdotp_Gampro_F2}) we have
\Beq
&&M(h)\Gamma^{(2)}_i M(h)^\dag=\sum_j D^{(\bar E_2)}_{ji}(h)\Gamma^{(2)}_j,\\
&&M(T_0)\Gamma^{(2)*}_i M(T_0)^\dag=\sum_j D^{(\bar E_2)}_{ji}(T_0)\Gamma^{(2)}_j.
\Eeq
The existence of quadratic dispersion terms with the form $\sum_{i=1}^2 \Gamma^{(2)}_i P_i^{(2)}(\delta\pmb k)$ can be judged using the formula (\ref{general_form_final}) by replacing the vector Rep $(v)$ with the linear Rep $(E_2)$. 

Applying the method introduced in section \ref{sec:disper_kpHam}, we can obtain the matrices $\Gamma^{(2)}_{1,2}$.

\section{$k\cdot p$ theory: hermiticity of $\Gamma^{1,2,3}$ }\label{app:condition}

We define the projection operator
\[
\mathbb P^{(I)} ={1\over|H|} \sum_{h\in H}  W(h)
\]
which project onto the subspace of identity Reps in the product Rep $W(h)=D^{(v)}(h)\otimes M(h)\otimes F(h)$. Similarly, 
\[
{1\over|H|} \sum_{h\in H}  W_{\eta_0}(h) =\mathbb P^{(I)}\mathbb P_{\eta_0}. 
\]
Following the discussion in Appendix \ref{app:Peta0}, the eigenvector $\tilde \Gamma$ of $\mathbb P^{(I)}\mathbb P_{\eta_0}$ with 
\[
\mathbb P^{(I)}\mathbb P_{\eta_0} \tilde \Gamma=\tilde \Gamma
\]
is also an eigenvector of $\mathbb P^{(I)}$ and $\mathbb P_{\eta_0}$, namely, it is a $\eta_0$-symmetric vector which carries the identity Rep of $H$.

Furthermore, referring to Appendix \ref{app:Peta0} and noticing the facts $W(T_0)K\mathbb P^{(I)}=\mathbb P^{(I)}W(T_0)K$ and $\mathbb T_{\eta_0} W(T_0)K \mathbb T_{\eta_0} = W(T_0)KW(\sigma)=W(T_0)K\mathbb T_{\eta_0}^2$, it can be verified that $W(T_0)K$ preserves the eigenspace of $\mathbb P^{(I)}$ and $\mathbb P_{\eta_0}$.
Namely, the eigenspace $\mathcal L_1$ of $\mathbb P^{(I)}\mathbb P_{\eta_0}$ with eigenvalue 1 is closed under the action of  any $g\in G$, hence forms a linear Rep space of $G$. From the lemma \ref{lemma2} and its corollary, we can choose proper bases such that each basis carries the identity Rep of $G$, namely, we can always find the bases of $\mathcal L_1$ such that $T_0$ is represented as $\mathscr M(T_0)K=IK$.\\

Now it is ready to prove that the matrices $\Gamma^{1,2,3}$ constructed from $\tilde\Gamma$ are hermitian matrices.\\

\begin{corollary}
If $\tilde\Gamma$ satisfy the relations (\ref{WGm}) and (\ref{Teta0}), then the resultant matrices $\Gamma^{1,2,3}$ are hermitian, where $\tilde\Gamma$ and $\Gamma^m$ are related by $(\tilde \Gamma)_{n\times d^2+i\times d+j} = (\tilde\Gamma^n)_{ij}$ and $\Gamma^m =   \tilde \Gamma^m M^*(T_0)$. 
\end{corollary}
\begin{proof} 
Equation (\ref{Teta0})  indicates that $\tilde\Gamma^m$ satisfies the symmetry condition
\begin{align}\label{app3:symEq}
(\tilde\Gamma^m)^T=\eta_0\sum_n D^{(v)}_{mn}(T_0)M(\sigma)\tilde\Gamma^n.
\end{align}
Taking complex conjugation, above equation becomes
\Beq
(\tilde \Gamma^m)^\dag &=& \sum_n \eta_0^*  D_{mn}^{(v)*}(T_0) M^*(\sigma)(\tilde \Gamma^n)^*\\
&=&M^*(T_0)M(T_0) \sum_{n} D_{mn}^{(v)*}(T_0) (\tilde \Gamma^n)^*,
\Eeq
which yields
\Beq
M^T(T_0)(\tilde \Gamma^m)^\dag =\sum_{n} D_{mn}^{(v)*}(T_0)M(T_0)  (\tilde \Gamma^n)^*.
\Eeq

On the other hand, the second equation in (\ref{WGm}) is equivalent to (\ref{tildeGaT}), which indicates that
\Beq
\tilde \Gamma^mM^*(T_0) =  \sum_{k}D^{(v)}_{mk}(T_0)M(T_0) (\tilde\Gamma^k)^*.
\Eeq
Since the vector Rep is a real Rep, comparing above two equations we have $$M^T(T_0)(\tilde \Gamma^m)^\dag = \tilde \Gamma^mM^*(T_0),$$ namely, $\big(\Gamma^m\big)^\dag=\Gamma^m$. This completes the proof. 
\end{proof}

\bibliography{Ham_CG_Ref}

\begin{thebibliography}{42}
\expandafter\ifx\csname natexlab\endcsname\relax\def\natexlab#1{#1}\fi
\expandafter\ifx\csname bibnamefont\endcsname\relax
  \def\bibnamefont#1{#1}\fi
\expandafter\ifx\csname bibfnamefont\endcsname\relax
  \def\bibfnamefont#1{#1}\fi
\expandafter\ifx\csname citenamefont\endcsname\relax
  \def\citenamefont#1{#1}\fi
\expandafter\ifx\csname url\endcsname\relax
  \def\url#1{\texttt{#1}}\fi
\expandafter\ifx\csname urlprefix\endcsname\relax\def\urlprefix{URL }\fi
\providecommand{\bibinfo}[2]{#2}
\providecommand{\eprint}[2][]{\url{#2}}

\bibitem[{\citenamefont{Schur}(01 Jan. 1904)}]{Schur}
\bibinfo{author}{\bibfnamefont{J.}~\bibnamefont{Schur}},
  \bibinfo{journal}{Journal für die reine und angewandte Mathematik}
  \textbf{\bibinfo{volume}{1904}}, \bibinfo{pages}{20 } (\bibinfo{year}{01 Jan.
  1904}),
  \urlprefix\url{https://www.degruyter.com/view/journals/crll/1904/127/article-p20.xml}.

\bibitem[{\citenamefont{Pollmann et~al.}(2010)\citenamefont{Pollmann, Turner,
  Berg, and Oshikawa}}]{PRB2}
\bibinfo{author}{\bibfnamefont{F.}~\bibnamefont{Pollmann}},
  \bibinfo{author}{\bibfnamefont{A.~M.} \bibnamefont{Turner}},
  \bibinfo{author}{\bibfnamefont{E.}~\bibnamefont{Berg}}, \bibnamefont{and}
  \bibinfo{author}{\bibfnamefont{M.}~\bibnamefont{Oshikawa}},
  \bibinfo{journal}{Phys. Rev. B} \textbf{\bibinfo{volume}{81}},
  \bibinfo{pages}{064439} (\bibinfo{year}{2010}),
  \urlprefix\url{https://link.aps.org/doi/10.1103/PhysRevB.81.064439}.

\bibitem[{\citenamefont{Chen et~al.}(2011{\natexlab{a}})\citenamefont{Chen, Gu,
  and Wen}}]{PRB3}
\bibinfo{author}{\bibfnamefont{X.}~\bibnamefont{Chen}},
  \bibinfo{author}{\bibfnamefont{Z.-C.} \bibnamefont{Gu}}, \bibnamefont{and}
  \bibinfo{author}{\bibfnamefont{X.-G.} \bibnamefont{Wen}},
  \bibinfo{journal}{Phys. Rev. B} \textbf{\bibinfo{volume}{83}},
  \bibinfo{pages}{035107} (\bibinfo{year}{2011}{\natexlab{a}}),
  \urlprefix\url{https://link.aps.org/doi/10.1103/PhysRevB.83.035107}.

\bibitem[{\citenamefont{Chen et~al.}(2011{\natexlab{b}})\citenamefont{Chen, Gu,
  and Wen}}]{PRB4}
\bibinfo{author}{\bibfnamefont{X.}~\bibnamefont{Chen}},
  \bibinfo{author}{\bibfnamefont{Z.-C.} \bibnamefont{Gu}}, \bibnamefont{and}
  \bibinfo{author}{\bibfnamefont{X.-G.} \bibnamefont{Wen}},
  \bibinfo{journal}{Phys. Rev. B} \textbf{\bibinfo{volume}{84}},
  \bibinfo{pages}{235128} (\bibinfo{year}{2011}{\natexlab{b}}),
  \urlprefix\url{https://link.aps.org/doi/10.1103/PhysRevB.84.235128}.

\bibitem[{\citenamefont{Chen et~al.}(2013)\citenamefont{Chen, Gu, Liu, and
  Wen}}]{PRB5}
\bibinfo{author}{\bibfnamefont{X.}~\bibnamefont{Chen}},
  \bibinfo{author}{\bibfnamefont{Z.-C.} \bibnamefont{Gu}},
  \bibinfo{author}{\bibfnamefont{Z.-X.} \bibnamefont{Liu}}, \bibnamefont{and}
  \bibinfo{author}{\bibfnamefont{X.-G.} \bibnamefont{Wen}},
  \bibinfo{journal}{Phys. Rev. B} \textbf{\bibinfo{volume}{87}},
  \bibinfo{pages}{155114} (\bibinfo{year}{2013}),
  \urlprefix\url{https://link.aps.org/doi/10.1103/PhysRevB.87.155114}.

\bibitem[{\citenamefont{Slager et~al.}(2013)\citenamefont{Slager, Mesaros,
  Juri{\v c}i{\'c}, and Zaanen}}]{Slager2013}
\bibinfo{author}{\bibfnamefont{R.-J.} \bibnamefont{Slager}},
  \bibinfo{author}{\bibfnamefont{A.}~\bibnamefont{Mesaros}},
  \bibinfo{author}{\bibfnamefont{V.}~\bibnamefont{Juri{\v c}i{\'c}}},
  \bibnamefont{and} \bibinfo{author}{\bibfnamefont{J.}~\bibnamefont{Zaanen}},
  \bibinfo{journal}{Nature Physics} \textbf{\bibinfo{volume}{9}},
  \bibinfo{pages}{98} (\bibinfo{year}{2013}),
  \urlprefix\url{https://doi.org/10.1038/nphys2513}.

\bibitem[{\citenamefont{Barkeshli et~al.}(2019)\citenamefont{Barkeshli,
  Bonderson, Cheng, and Wang}}]{Barkeshli_2019}
\bibinfo{author}{\bibfnamefont{M.}~\bibnamefont{Barkeshli}},
  \bibinfo{author}{\bibfnamefont{P.}~\bibnamefont{Bonderson}},
  \bibinfo{author}{\bibfnamefont{M.}~\bibnamefont{Cheng}}, \bibnamefont{and}
  \bibinfo{author}{\bibfnamefont{Z.}~\bibnamefont{Wang}},
  \bibinfo{journal}{Physical Review B} \textbf{\bibinfo{volume}{100}}
  (\bibinfo{year}{2019}), ISSN \bibinfo{issn}{2469-9969},
  \urlprefix\url{http://dx.doi.org/10.1103/PhysRevB.100.115147}.

\bibitem[{\citenamefont{Hamermesh}(1989)}]{GroupAppPhy}
\bibinfo{author}{\bibfnamefont{M.}~\bibnamefont{Hamermesh}},
  \emph{\bibinfo{title}{Group theory and its application to physical
  problems}}, Dover Books on Physics and Chemistry (\bibinfo{publisher}{Dover
  Publications}, \bibinfo{year}{1989}), ISBN
  \bibinfo{isbn}{9780486661810,0486661814},
  \urlprefix\url{http://gen.lib.rus.ec/book/index.php?md5=a86ba99acebcb82bc840426fab1b8efa}.

\bibitem[{\citenamefont{Chen et~al.}(1985)\citenamefont{Chen, Gao, and
  Ma}}]{Chen}
\bibinfo{author}{\bibfnamefont{J.-Q.} \bibnamefont{Chen}},
  \bibinfo{author}{\bibfnamefont{M.-J.} \bibnamefont{Gao}}, \bibnamefont{and}
  \bibinfo{author}{\bibfnamefont{G.-Q.} \bibnamefont{Ma}},
  \bibinfo{journal}{Rev. Mod. Phys.} \textbf{\bibinfo{volume}{57}},
  \bibinfo{pages}{211} (\bibinfo{year}{1985}),
  \urlprefix\url{https://link.aps.org/doi/10.1103/RevModPhys.57.211}.

\bibitem[{\citenamefont{Hasan and Kane}(2010)}]{RMPKane}
\bibinfo{author}{\bibfnamefont{M.~Z.} \bibnamefont{Hasan}} \bibnamefont{and}
  \bibinfo{author}{\bibfnamefont{C.~L.} \bibnamefont{Kane}},
  \bibinfo{journal}{Rev. Mod. Phys.} \textbf{\bibinfo{volume}{82}},
  \bibinfo{pages}{3045} (\bibinfo{year}{2010}),
  \urlprefix\url{https://link.aps.org/doi/10.1103/RevModPhys.82.3045}.

\bibitem[{\citenamefont{Qi and Zhang}(2011)}]{RMPZhang}
\bibinfo{author}{\bibfnamefont{X.-L.} \bibnamefont{Qi}} \bibnamefont{and}
  \bibinfo{author}{\bibfnamefont{S.-C.} \bibnamefont{Zhang}},
  \bibinfo{journal}{Rev. Mod. Phys.} \textbf{\bibinfo{volume}{83}},
  \bibinfo{pages}{1057} (\bibinfo{year}{2011}),
  \urlprefix\url{https://link.aps.org/doi/10.1103/RevModPhys.83.1057}.

\bibitem[{\citenamefont{Read and Green}(2000)}]{PRBRead}
\bibinfo{author}{\bibfnamefont{N.}~\bibnamefont{Read}} \bibnamefont{and}
  \bibinfo{author}{\bibfnamefont{D.}~\bibnamefont{Green}},
  \bibinfo{journal}{Phys. Rev. B} \textbf{\bibinfo{volume}{61}},
  \bibinfo{pages}{10267} (\bibinfo{year}{2000}),
  \urlprefix\url{https://link.aps.org/doi/10.1103/PhysRevB.61.10267}.

\bibitem[{\citenamefont{Qi et~al.}(2009)\citenamefont{Qi, Hughes, Raghu, and
  Zhang}}]{PRLZhang}
\bibinfo{author}{\bibfnamefont{X.-L.} \bibnamefont{Qi}},
  \bibinfo{author}{\bibfnamefont{T.~L.} \bibnamefont{Hughes}},
  \bibinfo{author}{\bibfnamefont{S.}~\bibnamefont{Raghu}}, \bibnamefont{and}
  \bibinfo{author}{\bibfnamefont{S.-C.} \bibnamefont{Zhang}},
  \bibinfo{journal}{Phys. Rev. Lett.} \textbf{\bibinfo{volume}{102}},
  \bibinfo{pages}{187001} (\bibinfo{year}{2009}),
  \urlprefix\url{https://link.aps.org/doi/10.1103/PhysRevLett.102.187001}.

\bibitem[{\citenamefont{Christopher~Bradley}(2010)}]{MathInSolids}
\bibinfo{author}{\bibfnamefont{A.~C.} \bibnamefont{Christopher~Bradley}},
  \emph{\bibinfo{title}{The Mathematical Theory of Symmetry in Solids:
  Representation Theory for Point Groups and Space Groups}}, Oxford Classic
  Texts in the Physical Sciences (\bibinfo{publisher}{Oxford University Press},
  \bibinfo{year}{2010}), \bibinfo{edition}{1st} ed., ISBN
  \bibinfo{isbn}{0199582580,9780199582587},
  \urlprefix\url{http://gen.lib.rus.ec/book/index.php?md5=8aacfeeaa822a18b43c48f8a7973a304}.

\bibitem[{\citenamefont{Shaw and Lever}(1974)}]{CMP}
\bibinfo{author}{\bibfnamefont{R.}~\bibnamefont{Shaw}} \bibnamefont{and}
  \bibinfo{author}{\bibfnamefont{J.}~\bibnamefont{Lever}},
  \bibinfo{journal}{Communications in Mathematical Physics}
  \textbf{\bibinfo{volume}{38}}, \bibinfo{pages}{257} (\bibinfo{year}{1974}),
  \urlprefix\url{https://doi.org/10.1007/BF01607948}.

\bibitem[{\citenamefont{Kim}(1984)}]{JMP}
\bibinfo{author}{\bibfnamefont{S.~K.} \bibnamefont{Kim}},
  \bibinfo{journal}{Journal of Mathematical Physics}
  \textbf{\bibinfo{volume}{25}}, \bibinfo{pages}{197} (\bibinfo{year}{1984}),
  \eprint{https://doi.org/10.1063/1.526139},
  \urlprefix\url{https://doi.org/10.1063/1.526139}.

\bibitem[{\citenamefont{Bardeen}(1938)}]{Bardeen}
\bibinfo{author}{\bibfnamefont{J.}~\bibnamefont{Bardeen}},
  \bibinfo{journal}{The Journal of Chemical Physics}
  \textbf{\bibinfo{volume}{6}}, \bibinfo{pages}{367} (\bibinfo{year}{1938}),
  \eprint{https://doi.org/10.1063/1.1750270},
  \urlprefix\url{https://doi.org/10.1063/1.1750270}.

\bibitem[{\citenamefont{F.}(1940)}]{Seitzkp}
\bibinfo{author}{\bibfnamefont{S.}~\bibnamefont{F.}},
  \emph{\bibinfo{title}{Modern Theory of Solids (1987)(en)(736s)}},
  INTERNATIONAL SERIES IN PHYSICS (\bibinfo{publisher}{McGraw-Hill},
  \bibinfo{year}{1940}),
  \urlprefix\url{http://gen.lib.rus.ec/book/index.php?md5=8d029cfbf16e2e1f01e0d9c3f4ef80b8}.

\bibitem[{\citenamefont{Sakata}(1974)}]{JMP2}
\bibinfo{author}{\bibfnamefont{I.}~\bibnamefont{Sakata}},
  \bibinfo{journal}{Journal of Mathematical Physics}
  \textbf{\bibinfo{volume}{15}}, \bibinfo{pages}{1702} (\bibinfo{year}{1974}),
  \eprint{https://doi.org/10.1063/1.1666528},
  \urlprefix\url{https://doi.org/10.1063/1.1666528}.

\bibitem[{\citenamefont{Dirl}(1979)}]{JMP4}
\bibinfo{author}{\bibfnamefont{R.}~\bibnamefont{Dirl}},
  \bibinfo{journal}{Journal of Mathematical Physics}
  \textbf{\bibinfo{volume}{20}}, \bibinfo{pages}{659} (\bibinfo{year}{1979}),
  \eprint{https://doi.org/10.1063/1.524107},
  \urlprefix\url{https://doi.org/10.1063/1.524107}.

\bibitem[{\citenamefont{Yang and Liu}(2017)}]{YangLiu}
\bibinfo{author}{\bibfnamefont{J.}~\bibnamefont{Yang}} \bibnamefont{and}
  \bibinfo{author}{\bibfnamefont{Z.-X.} \bibnamefont{Liu}},
  \bibinfo{journal}{Journal of Physics A: Mathematical and Theoretical}
  \textbf{\bibinfo{volume}{51}}, \bibinfo{pages}{025207}
  (\bibinfo{year}{2017}),
  \urlprefix\url{https://doi.org/10.1088/1751-8121/aa971a}.

\bibitem[{\citenamefont{Young et~al.}(2012)\citenamefont{Young, Zaheer, Teo,
  Kane, Mele, and Rappe}}]{KaneDirac}
\bibinfo{author}{\bibfnamefont{S.~M.} \bibnamefont{Young}},
  \bibinfo{author}{\bibfnamefont{S.}~\bibnamefont{Zaheer}},
  \bibinfo{author}{\bibfnamefont{J.~C.~Y.} \bibnamefont{Teo}},
  \bibinfo{author}{\bibfnamefont{C.~L.} \bibnamefont{Kane}},
  \bibinfo{author}{\bibfnamefont{E.~J.} \bibnamefont{Mele}}, \bibnamefont{and}
  \bibinfo{author}{\bibfnamefont{A.~M.} \bibnamefont{Rappe}},
  \bibinfo{journal}{Phys. Rev. Lett.} \textbf{\bibinfo{volume}{108}},
  \bibinfo{pages}{140405} (\bibinfo{year}{2012}),
  \urlprefix\url{https://link.aps.org/doi/10.1103/PhysRevLett.108.140405}.

\bibitem[{\citenamefont{Tang et~al.}(2016)\citenamefont{Tang, Zhou, Xu, and
  Zhang}}]{Tang}
\bibinfo{author}{\bibfnamefont{P.}~\bibnamefont{Tang}},
  \bibinfo{author}{\bibfnamefont{Q.}~\bibnamefont{Zhou}},
  \bibinfo{author}{\bibfnamefont{G.}~\bibnamefont{Xu}}, \bibnamefont{and}
  \bibinfo{author}{\bibfnamefont{S.-C.} \bibnamefont{Zhang}},
  \bibinfo{journal}{Nature Physics} \textbf{\bibinfo{volume}{12}},
  \bibinfo{pages}{1100} (\bibinfo{year}{2016}),
  \urlprefix\url{https://doi.org/10.1038/nphys3839}.

\bibitem[{\citenamefont{Hua et~al.}(2018{\natexlab{a}})\citenamefont{Hua, Nie,
  Song, Yu, Xu, and Yao}}]{XuPRB}
\bibinfo{author}{\bibfnamefont{G.}~\bibnamefont{Hua}},
  \bibinfo{author}{\bibfnamefont{S.}~\bibnamefont{Nie}},
  \bibinfo{author}{\bibfnamefont{Z.}~\bibnamefont{Song}},
  \bibinfo{author}{\bibfnamefont{R.}~\bibnamefont{Yu}},
  \bibinfo{author}{\bibfnamefont{G.}~\bibnamefont{Xu}}, \bibnamefont{and}
  \bibinfo{author}{\bibfnamefont{K.}~\bibnamefont{Yao}},
  \bibinfo{journal}{Phys. Rev. B} \textbf{\bibinfo{volume}{98}},
  \bibinfo{pages}{201116} (\bibinfo{year}{2018}{\natexlab{a}}),
  \urlprefix\url{https://link.aps.org/doi/10.1103/PhysRevB.98.201116}.

\bibitem[{\citenamefont{Armitage et~al.}(2018)\citenamefont{Armitage, Mele, and
  Vishwanath}}]{Armitage}
\bibinfo{author}{\bibfnamefont{N.~P.} \bibnamefont{Armitage}},
  \bibinfo{author}{\bibfnamefont{E.~J.} \bibnamefont{Mele}}, \bibnamefont{and}
  \bibinfo{author}{\bibfnamefont{A.}~\bibnamefont{Vishwanath}},
  \bibinfo{journal}{Rev. Mod. Phys.} \textbf{\bibinfo{volume}{90}},
  \bibinfo{pages}{015001} (\bibinfo{year}{2018}),
  \urlprefix\url{https://link.aps.org/doi/10.1103/RevModPhys.90.015001}.

\bibitem[{\citenamefont{Hua et~al.}(2018{\natexlab{b}})\citenamefont{Hua, Nie,
  Song, Yu, Xu, and Yao}}]{HUA}
\bibinfo{author}{\bibfnamefont{G.}~\bibnamefont{Hua}},
  \bibinfo{author}{\bibfnamefont{S.}~\bibnamefont{Nie}},
  \bibinfo{author}{\bibfnamefont{Z.}~\bibnamefont{Song}},
  \bibinfo{author}{\bibfnamefont{R.}~\bibnamefont{Yu}},
  \bibinfo{author}{\bibfnamefont{G.}~\bibnamefont{Xu}}, \bibnamefont{and}
  \bibinfo{author}{\bibfnamefont{K.}~\bibnamefont{Yao}},
  \bibinfo{journal}{Phys. Rev. B} \textbf{\bibinfo{volume}{98}},
  \bibinfo{pages}{201116} (\bibinfo{year}{2018}{\natexlab{b}}),
  \urlprefix\url{https://link.aps.org/doi/10.1103/PhysRevB.98.201116}.

\bibitem[{\citenamefont{Watanabe et~al.}(2018)\citenamefont{Watanabe, Po, and
  Vishwanath}}]{Watan}
\bibinfo{author}{\bibfnamefont{H.}~\bibnamefont{Watanabe}},
  \bibinfo{author}{\bibfnamefont{H.~C.} \bibnamefont{Po}}, \bibnamefont{and}
  \bibinfo{author}{\bibfnamefont{A.}~\bibnamefont{Vishwanath}},
  \bibinfo{journal}{Science Advances} \textbf{\bibinfo{volume}{4}}
  (\bibinfo{year}{2018}),
  \urlprefix\url{https://advances.sciencemag.org/content/4/8/eaat8685}.

\bibitem[{\citenamefont{Cano et~al.}(2019)\citenamefont{Cano, Bradlyn, and
  Vergniory}}]{Cano}
\bibinfo{author}{\bibfnamefont{J.}~\bibnamefont{Cano}},
  \bibinfo{author}{\bibfnamefont{B.}~\bibnamefont{Bradlyn}}, \bibnamefont{and}
  \bibinfo{author}{\bibfnamefont{M.~G.} \bibnamefont{Vergniory}},
  \bibinfo{journal}{APL Materials} \textbf{\bibinfo{volume}{7}},
  \bibinfo{pages}{101125} (\bibinfo{year}{2019}),
  \eprint{https://doi.org/10.1063/1.5124314},
  \urlprefix\url{https://doi.org/10.1063/1.5124314}.

\bibitem[{\citenamefont{Xu et~al.}(2020)\citenamefont{Xu, Elcoro, Song, Wieder,
  Vergniory, Regnault, Chen, Felser, and Bernevig}}]{XuNature}
\bibinfo{author}{\bibfnamefont{Y.}~\bibnamefont{Xu}},
  \bibinfo{author}{\bibfnamefont{L.}~\bibnamefont{Elcoro}},
  \bibinfo{author}{\bibfnamefont{Z.-D.} \bibnamefont{Song}},
  \bibinfo{author}{\bibfnamefont{B.~J.} \bibnamefont{Wieder}},
  \bibinfo{author}{\bibfnamefont{M.~G.} \bibnamefont{Vergniory}},
  \bibinfo{author}{\bibfnamefont{N.}~\bibnamefont{Regnault}},
  \bibinfo{author}{\bibfnamefont{Y.}~\bibnamefont{Chen}},
  \bibinfo{author}{\bibfnamefont{C.}~\bibnamefont{Felser}}, \bibnamefont{and}
  \bibinfo{author}{\bibfnamefont{B.~A.} \bibnamefont{Bernevig}},
  \bibinfo{journal}{Nature} \textbf{\bibinfo{volume}{586}},
  \bibinfo{pages}{702} (\bibinfo{year}{2020}),
  \urlprefix\url{https://doi.org/10.1038/s41586-020-2837-0}.

\bibitem[{\citenamefont{Elcoro et~al.}(2020)\citenamefont{Elcoro, Wieder, Song,
  Xu, Bradlyn, and Bernevig}}]{elcoro2020magnetic}
\bibinfo{author}{\bibfnamefont{L.}~\bibnamefont{Elcoro}},
  \bibinfo{author}{\bibfnamefont{B.~J.} \bibnamefont{Wieder}},
  \bibinfo{author}{\bibfnamefont{Z.}~\bibnamefont{Song}},
  \bibinfo{author}{\bibfnamefont{Y.}~\bibnamefont{Xu}},
  \bibinfo{author}{\bibfnamefont{B.}~\bibnamefont{Bradlyn}}, \bibnamefont{and}
  \bibinfo{author}{\bibfnamefont{B.~A.} \bibnamefont{Bernevig}},
  \emph{\bibinfo{title}{Magnetic topological quantum chemistry}}
  (\bibinfo{year}{2020}), \eprint{2010.00598}.

\bibitem[{\citenamefont{Bouhon et~al.}(2020)\citenamefont{Bouhon, Lange, and
  Slager}}]{Bouhon2020}
\bibinfo{author}{\bibfnamefont{A.}~\bibnamefont{Bouhon}},
  \bibinfo{author}{\bibfnamefont{G.~F.} \bibnamefont{Lange}}, \bibnamefont{and}
  \bibinfo{author}{\bibfnamefont{R.-J.} \bibnamefont{Slager}},
  \emph{\bibinfo{title}{Topological correspondence between magnetic space group
  representations}} (\bibinfo{year}{2020}), \eprint{2010.10536}.

\bibitem[{\citenamefont{Burkov et~al.}(2011)\citenamefont{Burkov, Hook, and
  Balents}}]{NodalLine1}
\bibinfo{author}{\bibfnamefont{A.~A.} \bibnamefont{Burkov}},
  \bibinfo{author}{\bibfnamefont{M.~D.} \bibnamefont{Hook}}, \bibnamefont{and}
  \bibinfo{author}{\bibfnamefont{L.}~\bibnamefont{Balents}},
  \bibinfo{journal}{Phys. Rev. B} \textbf{\bibinfo{volume}{84}},
  \bibinfo{pages}{235126} (\bibinfo{year}{2011}),
  \urlprefix\url{https://link.aps.org/doi/10.1103/PhysRevB.84.235126}.

\bibitem[{\citenamefont{Burkov and Balents}(2011)}]{Burkov2011}
\bibinfo{author}{\bibfnamefont{A.~A.} \bibnamefont{Burkov}} \bibnamefont{and}
  \bibinfo{author}{\bibfnamefont{L.}~\bibnamefont{Balents}},
  \bibinfo{journal}{Phys. Rev. Lett.} \textbf{\bibinfo{volume}{107}},
  \bibinfo{pages}{127205} (\bibinfo{year}{2011}),
  \urlprefix\url{https://link.aps.org/doi/10.1103/PhysRevLett.107.127205}.

\bibitem[{\citenamefont{Xu et~al.}(2011)\citenamefont{Xu, Weng, Wang, Dai, and
  Fang}}]{NodalLine3}
\bibinfo{author}{\bibfnamefont{G.}~\bibnamefont{Xu}},
  \bibinfo{author}{\bibfnamefont{H.}~\bibnamefont{Weng}},
  \bibinfo{author}{\bibfnamefont{Z.}~\bibnamefont{Wang}},
  \bibinfo{author}{\bibfnamefont{X.}~\bibnamefont{Dai}}, \bibnamefont{and}
  \bibinfo{author}{\bibfnamefont{Z.}~\bibnamefont{Fang}},
  \bibinfo{journal}{Phys. Rev. Lett.} \textbf{\bibinfo{volume}{107}},
  \bibinfo{pages}{186806} (\bibinfo{year}{2011}),
  \urlprefix\url{https://link.aps.org/doi/10.1103/PhysRevLett.107.186806}.

\bibitem[{\citenamefont{Chen et~al.}(2015)\citenamefont{Chen, Xie, Yang, Pan,
  Zhang, Cohen, and Zhang}}]{NodalLine2}
\bibinfo{author}{\bibfnamefont{Y.}~\bibnamefont{Chen}},
  \bibinfo{author}{\bibfnamefont{Y.}~\bibnamefont{Xie}},
  \bibinfo{author}{\bibfnamefont{S.~A.} \bibnamefont{Yang}},
  \bibinfo{author}{\bibfnamefont{H.}~\bibnamefont{Pan}},
  \bibinfo{author}{\bibfnamefont{F.}~\bibnamefont{Zhang}},
  \bibinfo{author}{\bibfnamefont{M.~L.} \bibnamefont{Cohen}}, \bibnamefont{and}
  \bibinfo{author}{\bibfnamefont{S.}~\bibnamefont{Zhang}},
  \bibinfo{journal}{Nano Letters} \textbf{\bibinfo{volume}{15}},
  \bibinfo{pages}{6974} (\bibinfo{year}{2015}), \bibinfo{note}{pMID: 26426355},
  \eprint{https://doi.org/10.1021/acs.nanolett.5b02978},
  \urlprefix\url{https://doi.org/10.1021/acs.nanolett.5b02978}.

\bibitem[{\citenamefont{Fang et~al.}(2015)\citenamefont{Fang, Chen, Kee, and
  Fu}}]{Fang2015}
\bibinfo{author}{\bibfnamefont{C.}~\bibnamefont{Fang}},
  \bibinfo{author}{\bibfnamefont{Y.}~\bibnamefont{Chen}},
  \bibinfo{author}{\bibfnamefont{H.-Y.} \bibnamefont{Kee}}, \bibnamefont{and}
  \bibinfo{author}{\bibfnamefont{L.}~\bibnamefont{Fu}}, \bibinfo{journal}{Phys.
  Rev. B} \textbf{\bibinfo{volume}{92}}, \bibinfo{pages}{081201}
  (\bibinfo{year}{2015}),
  \urlprefix\url{https://link.aps.org/doi/10.1103/PhysRevB.92.081201}.

\bibitem[{\citenamefont{Weng et~al.}(2015)\citenamefont{Weng, Liang, Xu, Yu,
  Fang, Dai, and Kawazoe}}]{Weng2015}
\bibinfo{author}{\bibfnamefont{H.}~\bibnamefont{Weng}},
  \bibinfo{author}{\bibfnamefont{Y.}~\bibnamefont{Liang}},
  \bibinfo{author}{\bibfnamefont{Q.}~\bibnamefont{Xu}},
  \bibinfo{author}{\bibfnamefont{R.}~\bibnamefont{Yu}},
  \bibinfo{author}{\bibfnamefont{Z.}~\bibnamefont{Fang}},
  \bibinfo{author}{\bibfnamefont{X.}~\bibnamefont{Dai}}, \bibnamefont{and}
  \bibinfo{author}{\bibfnamefont{Y.}~\bibnamefont{Kawazoe}},
  \bibinfo{journal}{Phys. Rev. B} \textbf{\bibinfo{volume}{92}},
  \bibinfo{pages}{045108} (\bibinfo{year}{2015}),
  \urlprefix\url{https://link.aps.org/doi/10.1103/PhysRevB.92.045108}.

\bibitem[{\citenamefont{Bzdušek et~al.}(2016)\citenamefont{Bzdušek, Wu,
  Rüegg, Sigrist, and Soluyanov}}]{NodalLine4}
\bibinfo{author}{\bibfnamefont{T.}~\bibnamefont{Bzdušek}},
  \bibinfo{author}{\bibfnamefont{Q.}~\bibnamefont{Wu}},
  \bibinfo{author}{\bibfnamefont{A.}~\bibnamefont{Rüegg}},
  \bibinfo{author}{\bibfnamefont{M.}~\bibnamefont{Sigrist}}, \bibnamefont{and}
  \bibinfo{author}{\bibfnamefont{A.~A.} \bibnamefont{Soluyanov}},
  \bibinfo{journal}{Nature} \textbf{\bibinfo{volume}{538}},
  \bibinfo{pages}{75–78} (\bibinfo{year}{2016}), ISSN
  \bibinfo{issn}{1476-4687},
  \urlprefix\url{http://dx.doi.org/10.1038/nature19099}.

\bibitem[{\citenamefont{Geilhufe et~al.}(2019)\citenamefont{Geilhufe, Guinea,
  and Juri\ifmmode \check{c}\else \v{c}\fi{}i\ifmmode~\acute{c}\else
  \'{c}\fi{}}}]{Gei2019}
\bibinfo{author}{\bibfnamefont{R.~M.} \bibnamefont{Geilhufe}},
  \bibinfo{author}{\bibfnamefont{F.}~\bibnamefont{Guinea}}, \bibnamefont{and}
  \bibinfo{author}{\bibfnamefont{V.}~\bibnamefont{Juri\ifmmode \check{c}\else
  \v{c}\fi{}i\ifmmode~\acute{c}\else \'{c}\fi{}}}, \bibinfo{journal}{Phys. Rev.
  B} \textbf{\bibinfo{volume}{99}}, \bibinfo{pages}{020404}
  (\bibinfo{year}{2019}),
  \urlprefix\url{https://link.aps.org/doi/10.1103/PhysRevB.99.020404}.

\bibitem[{\citenamefont{Guo et~al.}()\citenamefont{Guo, Guo, Tan, Feng, Cao,
  Liu, Liu, Lu, and Ji}}]{guo2020}
\bibinfo{author}{\bibfnamefont{D.}~\bibnamefont{Guo}},
  \bibinfo{author}{\bibfnamefont{P.}~\bibnamefont{Guo}},
  \bibinfo{author}{\bibfnamefont{S.}~\bibnamefont{Tan}},
  \bibinfo{author}{\bibfnamefont{M.}~\bibnamefont{Feng}},
  \bibinfo{author}{\bibfnamefont{L.}~\bibnamefont{Cao}},
  \bibinfo{author}{\bibfnamefont{Z.}~\bibnamefont{Liu}},
  \bibinfo{author}{\bibfnamefont{K.}~\bibnamefont{Liu}},
  \bibinfo{author}{\bibfnamefont{Z.-Y.} \bibnamefont{Lu}}, \bibnamefont{and}
  \bibinfo{author}{\bibfnamefont{W.}~\bibnamefont{Ji}},
  \eprint{arXiv:2012.15218}.

\bibitem[{\citenamefont{Cui et~al.}()\citenamefont{Cui, Li, Guo, Guo, Lou, Mei,
  Lin, Tan, Liu, Liu et~al.}}]{cui2020}
\bibinfo{author}{\bibfnamefont{X.}~\bibnamefont{Cui}},
  \bibinfo{author}{\bibfnamefont{Y.}~\bibnamefont{Li}},
  \bibinfo{author}{\bibfnamefont{D.}~\bibnamefont{Guo}},
  \bibinfo{author}{\bibfnamefont{P.}~\bibnamefont{Guo}},
  \bibinfo{author}{\bibfnamefont{C.}~\bibnamefont{Lou}},
  \bibinfo{author}{\bibfnamefont{G.}~\bibnamefont{Mei}},
  \bibinfo{author}{\bibfnamefont{C.}~\bibnamefont{Lin}},
  \bibinfo{author}{\bibfnamefont{S.}~\bibnamefont{Tan}},
  \bibinfo{author}{\bibfnamefont{Z.}~\bibnamefont{Liu}},
  \bibinfo{author}{\bibfnamefont{K.}~\bibnamefont{Liu}}, \bibnamefont{et~al.},
  \eprint{arXiv:2012.15220}.

\bibitem[{\citenamefont{Yang et~al.}(2021)\citenamefont{Yang, Fang, and
  Liu}}]{YFL}
\bibinfo{author}{\bibfnamefont{J.}~\bibnamefont{Yang}},
  \bibinfo{author}{\bibfnamefont{C.}~\bibnamefont{Fang}}, \bibnamefont{and}
  \bibinfo{author}{\bibfnamefont{Z.-X.} \bibnamefont{Liu}}
  (\bibinfo{year}{2021}), \eprint{arXiv:2101.01733}.

\end{thebibliography}

\end{document}